\documentclass[lettersize,journal]{IEEEtran}
\IEEEoverridecommandlockouts
\usepackage{cite}
\usepackage{amsmath,amssymb,amsfonts,amsthm}
\usepackage{algorithmic}
\usepackage{algorithm}
\usepackage{array}
\usepackage{graphicx}
\usepackage{textcomp}
\usepackage{xcolor}
\usepackage{amsthm}
\usepackage{subfigure}
\usepackage{graphicx}
\usepackage{float}
\usepackage{enumerate}
\usepackage{tabularx}
\usepackage{stfloats}
\usepackage{url}
\usepackage{hyperref}
\usepackage{xcolor, soul, framed} 
\usepackage{xcolor}
\usepackage{verbatim}
\usepackage{pstricks}
\usepackage{booktabs} 
\usepackage{makecell} 
\graphicspath{{figure/}} 

\newtheorem{thm}{Theorem}

\newtheorem{lem}[thm]{Lemma}

\newtheorem{rem}[thm]{Remark}

\begin{document}
\renewcommand{\algorithmicrequire}{\textbf{Input:}} 
\renewcommand{\algorithmicensure}{\textbf{Output:}}
\title{Toward Wireless Localization Using Multiple Reconfigurable Intelligent Surfaces\\

\thanks{This work is supported by the National Natural Science Foundation of China under Grant No.12141107, and the Interdisciplinary Research Program of HUST (2023JCYJ012). (Corresponding author: Tiebin Mi)}
}

\author{Fuhai Wang,~\IEEEmembership{Student Member,~IEEE,}
Tiebin Mi,~\IEEEmembership{Member,~IEEE,}
Chun Wang, 
Rujing Xiong,~\IEEEmembership{Student Member,~IEEE,}
Zhengyu Wang,~\IEEEmembership{Student Member,~IEEE,} 
and
Robert Caiming Qiu,~\IEEEmembership{Fellow,~IEEE} 
\thanks{Fuhai Wang is with the School of Electronic Information and Communications, Huazhong University of Science and Technology, Wuhan 430074, China, and also with the Institute of Artificial Intelligence, Huazhong University of Science and Technology, Wuhan 430074, China (e-mail: wangfuhai@hust.edu.cn).}
\thanks{Tiebin Mi, Chun Wang, Rujing Xiong, Zhengyu Wang and Robert Caiming Qiu are with the School of Electronic Information and Communications, Huazhong University of Science and Technology, Wuhan 430074, China (e-mail: mitiebin@hust.edu.cn; m202272439@hust.edu.cn; rujing@hust.edu.cn; wangzhengyu@hust.edu.cn; caiming@hust.edu.cn).}
} 

\maketitle
\begin{abstract}
This paper investigates the capabilities and effectiveness of backward sensing centered on reconfigurable intelligent surfaces (RISs). We demonstrate that the direction of arrival (DoA) estimation of incident waves in the far-field regime can be accomplished using a single RIS by leveraging configurational diversity. Furthermore, we identify that the spatial diversity achieved through deploying multiple RISs enables accurate localization of multiple power sources. Physically accurate and mathematically concise models are introduced to characterize forward signal aggregations via RISs. By employing linearized approximations inherent in the far-field region, the measurement process for various configurations can be expressed as a system of linear equations. The mathematical essence of backward sensing lies in solving this system. A theoretical framework for determining key performance indicators is established through condition number analysis of the sensing operators. In the context of localization using multiple RISs, we examine relationships among the rank of sensing operators, the size of the region of interest (RoI), and the number of elements and measurements. For DoA estimations, we provide an upper bound for the relative error of the least squares reconstruction algorithm. These quantitative analyses offer essential insights for system design and optimization. Numerical experiments validate our findings. To demonstrate the practicality of our proposed RIS-centric sensing approach, we develop a proof-of-concept prototype using universal software radio peripherals (USRP) and employ a magnitude-only reconstruction algorithm tailored for this system. To our knowledge, this represents the first trial of its kind.
\end{abstract}

\begin{IEEEkeywords}
Reconfigurable Intelligent Surface (RIS), forward signal aggregations, backward sensing, localization, direction of arrival (DoA). 
\end{IEEEkeywords}


\section{Introduction} 

\IEEEPARstart{R}econfigurable intelligent surfaces (RISs) have received considerable attention due to their potential to enhance both the communication and sensing capabilities of wireless networks \cite{du2024nested,tang2020mimo,wu2019intelligent,basar2019wireless,song2023intelligent,zhong2023joint}. RISs comprise cost-effective, well-designed electromagnetic (EM) units, each capable of independently modifying the characteristics of incident EM waves. These units enable the control of the scattering and reflection properties of EM waves through software-based manipulation. Although current research extensively explores their applications in forward coverage enhancement under line-of-sight/non line-of-sight (LoS/NLoS) conditions \cite{WOS,xiong2024optimal,ma2022cooperative,chu2021intelligent,ouyang2023computer,cheng2022reconfigurable,wu2019intelligent,ma2022modeling}, backward sensing through RISs is also attracting increasing attention \cite{zhang2021metalocalization,lin2021single,he2022high,shao2022target}. A key advantage of RISs in sensing applications is their reconfigurability of the electromagnetic environment, which distinguishes them from other competitors~\cite{moghadasian2019sparsely,ahmed2021uwb,zhao2023tdloc} in the design of sensing signals with temporal diversity or frequency diversity.

The diversity schemes introduced by RISs are twofold: configurational diversity~\cite{lan2020wireless,alexandropoulos2023ris,zhao2023intelligent}, enabled by their reconfigurability, and spatial diversity, achieved through deploying multiple RISs. Current research primarily focuses on exploring the configurational diversity of an individual RIS. In these studies, the RIS is often regarded as a promising auxiliary technology for enhancing wireless systems by producing a sharp point spread function \cite{jiang2024near,buzzi2021radar,wymeersch2020radio,li2023metaphys,jiang2022reconfigurable}, which facilitates the generation of an unblurred image or clear map (e.g., \cite{keykhosravi2023leveraging,lin2021single}). Notably, RISs themselves are well-suited for sensing tasks due to their nature as large and reconfigurable antenna arrays. In general, RISs can generate a series of high-dimensional observations for regions of interest (RoIs)~\cite{wymeersch2022radio}, which greatly benefit sensing applications~\cite{mehrotra2022does}. Nevertheless, the complete functionality of RISs in this context remains unclear. It is crucial to assess the feasibility and capabilities of the RIS-centric backward sensing approach.

A fundamental function of sensing is localization, which involves extracting location-related information such as the direction of arrival (DoA), also known as the angle of arrival (AoA)~\cite{gezici2008survey,guvenc2009survey}. However, there are only a limited number of articles addressing DoA estimation through RISs \cite{rinchi2022compressive,lin2021single,zhou2022two}. To the best of our knowledge, most research is based on numerical simulations, and there are no prototype demonstrations of this capability. One reason is that accurate physics-based models of RISs in backward sensing are not yet well-developed. In this paper, we first demonstrate that DoA estimation of incident waves in the far-field can be performed using a single RIS, leveraging configurational diversity.

Even with successful DoA estimation using RISs, accurate localization remains challenging~\cite{sayed2005network}. One approach involves deploying multiple RISs~\cite{alexandropoulos2022localization,keykhosravi2023leveraging} to enhance the spatial diversity. A DoA parameter defines a straight line passing through the power source and the RIS, meaning that two DoA parameters are sufficient to locate a single source via triangulation, as shown in Fig.~\ref{F:Discretization-1:1}. However, when multiple sources radiate signals simultaneously, ambiguity arises, as illustrated in Fig.~\ref{F:Discretization-1:2}. To resolve this uncertainty, additional RISs positioned at different angles can be used. This approach is analogous to computed tomography, which reconstructs a field based on dense angular samplings \cite{wilson2010radio,karl2023foundations}. However, the cost of such samplings via RISs is impractical. Fully utilizing the spatial diversity provided by deploying multiple RISs remains an open question. Addressing these challenges in RIS-centric backward sensing requires a paradigm shift.

\begin{figure}[!htbp]
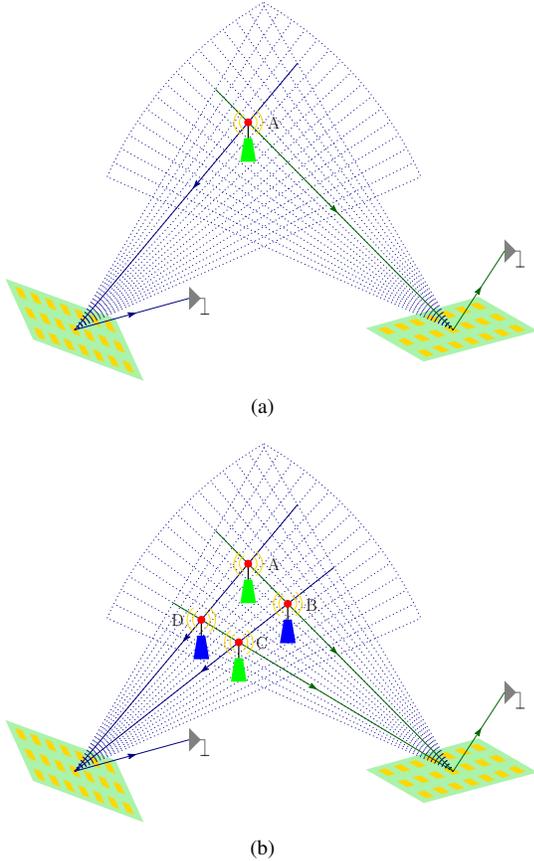

  \centering
  \subfigure[]{
  \label{F:Discretization-1:1}
  \includegraphics[width=.8\columnwidth]{figure/1-1-Discretization_1}}
  \subfigure[]{
  \label{F:Discretization-1:2}
  \includegraphics[width=.8\columnwidth]{figure/1-1-Discretization_2}}
  \caption{Illustration of the RIS-centric sensing system via triangulation: (a) A single source can be localized using two DoA parameters. (b) Ambiguity arises when multiple sources radiate signals simultaneously.}
  \label{F:Discretization-1} 
\end{figure}

\subsection{Contributions}
The main contributions are summarized as follows:

\begin{itemize}

\item \textbf{Physics-based forward signal aggregations and mathematical formulations for backward sensing:} 
Physically accurate and mathematically concise models are presented to characterize the reflection properties of RISs. By employing linearized approximations inherent in the far-field region, we provide practical models to describe the forward signal aggregation process via RISs. Additionally, we rewrite the measurements of the RoI for various configurations as a system of linear equations. This formulation offers valuable insights into the capabilities of the proposed RIS-centric sensing approach. Mathematically, the problem of backward sensing can be formulated as finding a reasonable solution to this system.

\item \textbf{Key indicators dominating the performance of backward sensing through RISs:} 
A theoretical framework to identify key performance indicators that dominate the performance of backward sensing is established through condition number analysis of the sensing operators. In the context of localization using multiple RISs, we investigate the relationships among the rank of sensing operators, the size of the RoI, and the number of elements and measurements. A necessary condition for high-fidelity recovery is that the dimension of the RoI is less than the number of elements and measurements. For DoA estimations, we provide an upper bound for the relative error of the least squares reconstruction algorithm. These quantitative analyses offer crucial insights for system design and optimization.

\item \textbf{Both numerical simulations and proof-of-concept prototypes demonstrate the feasibility of the proposed approach:}
We validate the capabilities and effectiveness of the proposed RIS-centric backward sensing approach through extensive experiments. Numerical trials illustrate the impact of several factors, including the number of elements, measurements, and RISs, as well as their positions. To demonstrate the practicality of our approach, we develop a proof-of-concept prototype using universal software radio peripherals (USRP) and employ a magnitude-only reconstruction algorithm tailored to this system. To the best of our knowledge, no previous prototypes have been proposed with such configurations. 

\end{itemize}

\subsection{Outline}
The remainder of the paper is organized as follows. In Section~\ref{S:ForwardAggregation}, we present forward signal aggregation models using the physical optics method. Section~\ref{S:BackwardSensing} addresses the RIS-centric backward sensing problem. Section~\ref{S:KeyIndicators} discusses key indicators dominating the performance of backward sensing. We validate the findings through extensive numerical experiments in Section~\ref{S:Experiments}. To demonstrate the practicality, we develop a proof-of-concept prototype using universal software radio peripherals (USRP) in Section~\ref{S:POC}. Finally, the paper is concluded in Section~\ref{S:Conclusion}.

\subsection{Reproducible Research}

Once the paper is accepted, we will release the code.

\subsection{Notations}
Unless explicitly specified, bold capital letters and bold small letters denote matrices and vectors, respectively. The conjugate transpose, transpose, and pseudo-inverse of $\mathbf{A}$ are denoted by~$\mathbf{A}^{*}$,~$\mathbf{A}^\top$ and~$\mathbf{A}^{\dag}$, respectively. We denote by $\| \cdot \|$ the Euclidean norm of a vector and the Frobenius norm of a matrix. The matrix $\text{diag}(\mathbf{a})$ represents a diagonal matrix with diagonal elements given by $\mathbf{a}$. Additionally, rank\{$\mathbf{A}$\} denotes the rank of the matrix $\mathbf{A}$.

\section{Forward Signal Aggregations via A Single RIS}\label{S:ForwardAggregation}

In the realm of RIS-centric inverse or backward sensing, a deep understanding of the forward measurement process is essential. Mathematically, this process is represented by the relation $S_t = H_t(\mathbf{E})$, where the scalar $S_t$ denotes the observation, the vector $\mathbf{E}$ represents the scene under investigation, and $H_t$ is the sensing operator mapping the scene to the observation. Here, we employ $t$ to index measurements corresponding to various configurations. For simplicity, we set the configuration to random sampling mode. Generally, RISs with random phase configurations are referred to as statistically isotropic RISs. While this configuration may not fully utilize the available degrees of freedom, it serves as a widely accepted benchmark for evaluating and comparing the performance of other RIS configurations \cite{mi2023towards,jiang2024near,sleasman2016microwave}.

\subsection{Linear Aggregations of Multiple Incident Signals}

We now examine the linear aggregation properties of RISs. Our analysis begins by concentrating on a single RIS situated in the $xoy$-plane, composed of multiple units located at $\mathbf{p}_n = [x_n, \ y_n, \ z_n ]^{\top}$, $n=1, \ldots, N$. For simplicity, we assume that the scattering pattern of an isolated unit is denoted by $\tau(\theta^\text{s}, \phi^\text{s}; \theta^\text{i}, \phi^\text{i})$. This notation signifies the dependence on both the incident angle $(\theta^\text{i}, \phi^\text{i})$ and the scattered angle $(\theta^\text{s}, \phi^\text{s})$. We will not delve into the specific form of $\tau(\cdot)$, as our goal is to illustrate the linearity of RISs. Using the physical optics method, we can calculate the bistatic scattered field of a rectangular metallic patch with near-zero thickness, as detailed in \cite{mi2023towards}. 

Let us consider an example illustrated in Fig.~\ref{F:ArrayGeometry}. The RIS is illuminated by a single incident EM wave originating from $(r^\text{i}, \theta^\text{i}, \phi^\text{i})$. Suppose the incident electric field arriving at $\mathbf{p}_n$ along $(\theta^\text{i} (n), \phi^\text{i} (n))$ is denoted by $E^\text{i} (n) $. Consequently, the scattered filed at $(\theta^\text{s} (n), \phi^\text{s} (n) )$ is represented as $E^\text{s} (n) = \tau ( \theta^\text{s} (n), \phi^\text{s} (n) ; \theta^\text{i} (n), \phi^\text{i} (n) ) E^\text{i} (n)$. In the case of a point source, radiation propagates radially. As the EM wave propagates toward the unit at $\mathbf{p}_n$, the attenuation behavior is characterized by the factor $r^\text{i} (n) e^{ j 2 \pi r^\text{i} (n) / \lambda }$. Therefore, $E^\text{i} (n) = E ( r^\text{i}, \theta^\text{i}, \phi^\text{i} ) e^{ - j 2 \pi r^\text{i} (n) / \lambda}/ r^\text{i} (n)$. Similarly, if $E^\text{s} (n)$ denotes the scattered electric field of the unit at $\mathbf{p}_n$, then the electric field at the observation point is $S_n = E^\text{s} (n) e^{ - j 2 \pi r^\text{s} (n) / \lambda } / r^\text{s} (n)$. 

\begin{figure}[!htbp]
  \centering
  \includegraphics[width=0.92\linewidth]{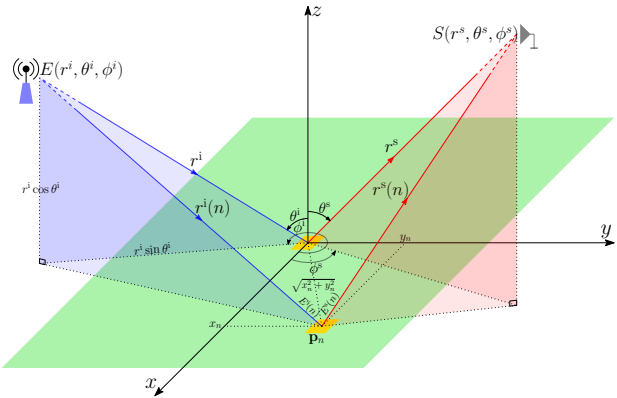}
  \caption{An RIS illuminated by a single incident EM wave from $(r^\text{i}, \theta^\text{i}, \phi^\text{i})$. The scattered waves from various elements are aggregated at $(r^\text{s}, \theta^\text{s}, \phi^\text{s})$.}
  \label{F:ArrayGeometry}
\end{figure}

Combining the three components along the propagation trace through the $n$-th unit at $\mathbf{p}_n$, the electric field at the observation point is expressed as 
\begin{multline}\label{SISOBehavior}
  S_n ( r^\text{s}_n, \theta^\text{s}, \phi^\text{s} ) = \tau_n ( \theta^\text{s} (n), \phi^\text{s} (n) ; \theta^\text{i} (n), \phi^\text{i} (n) ) e^{j \Omega_n} \\
  \frac{ e^{ - j 2 \pi r^\text{i} (n) / \lambda} e^{ - j 2 \pi r^\text{s} (n) / \lambda } }{ r^\text{i} (n) r^\text{s} (n) } E ( r^\text{i}, \theta^\text{i}, \phi^\text{i} ) . 
\end{multline}
Here, $\Omega_n$ represents the phase configuration of the $n$-th unit within the RIS. The aggregated electric field at the observation point is determined by the superposition of individual fields scattered by $N$ units, given by
\begin{multline}\label{Model_SISO}
  S ( r^\text{s}, \theta^\text{s}, \phi^\text{s} ) = E ( r^\text{i}, \theta^\text{i}, \phi^\text{i} ) \sum_{n=1}^{N} \tau_n ( \theta^\text{s} (n), \phi^\text{s} (n) ; \theta^\text{i} (n), \phi^\text{i} (n) ) \\ 
  e^{j \Omega_n} \frac{ e^{ - j 2 \pi r^\text{i} (n) / \lambda} e^{ - j 2 \pi r^\text{s} (n) / \lambda } }{ r^\text{i} (n)  r^\text{s} (n) }   . 
\end{multline}
When the RIS is exposed to multiple incident waves, the aggregated electric field is the superposition of the individual fields
\begin{equation}\label{Model_MISO}
\begin{aligned}
  S ( r^\text{s}, \theta^\text{s}, \phi^\text{s} ) = 
  & \sum_{m=1}^{M} E ( r^\text{i}_m, \theta^\text{i}_m, \phi^\text{i}_m )  \\
  & \sum_{n=1}^{N} \tau ( \theta^\text{s} (n), \phi^\text{s} (n) ; \theta^\text{i}_m (n), \phi^\text{i}_m (n) ) \\ 
  & e^{j \Omega_n} \frac{ e^{ - j 2 \pi r^\text{i}_m (n) / \lambda} e^{ - j 2 \pi r^\text{s} (n) / \lambda } }{ r^\text{i}_m (n)  r^\text{s} (n) }   . 
\end{aligned}
\end{equation}

\begin{figure*}[!htbp]
  \begin{equation}\label{E:MIMO}
    \begin{aligned}
        S (r^\text{s}, \theta^\text{s}, \phi^\text{s} ) 
      \approx & 
      \tau \frac{ e^{-j 2 \pi r^\text{s} / \lambda }}{ r^\text{s} } 
      \begin{bmatrix}
        e^{ j 2 \pi \mathbf{p}_1^{\top} \mathbf{u} (\theta^\text{s}, \phi^\text{s})  / \lambda } & \cdots & e^{ j 2 \pi \mathbf{p}_N^\top \mathbf{u} (\theta^\text{s}, \phi^\text{s})  / \lambda }  
      \end{bmatrix} 
      \begin{bmatrix}
        e^{j \Omega_1} &        &  0 \\
                               & \ddots &    \\
            0                  &        & e^{j \Omega_N}
      \end{bmatrix} \\
      & \begin{bmatrix}
        e^{ j 2 \pi \mathbf{p}_1^\top \mathbf{u} (\theta_1^{\text{i}}, \phi_1^\text{i})  / \lambda }  & \cdots  & e^{ j 2 \pi \mathbf{p}_1^\top \mathbf{u} (\theta_M^\text{i}, \phi_M^\text{i})  / \lambda } \\
        \vdots  & \ddots & \vdots \\
        e^{ j 2 \pi \mathbf{p}_N^\top \mathbf{u} (\theta_1^{\text{i}}, \phi_1^\text{i})  / \lambda }  & \cdots  & e^{ j 2 \pi \mathbf{p}_N^\top \mathbf{u} (\theta_M^\text{i}, \phi_M^\text{i})  / \lambda }\\
      \end{bmatrix}
      \begin{bmatrix}
        \frac{ e^{-j 2 \pi r^\text{i}_1 / \lambda } }{ r^\text{i}_1 } &        &  0 \\
                               & \ddots &    \\
            0                  &        &  \frac{ e^{-j 2 \pi r^\text{i}_M / \lambda } }{ r^\text{i}_M }
      \end{bmatrix}
      \begin{bmatrix}
        E (r_1^\text{i}, \theta_{1}^\text{i}, \phi_{1}^\text{i}) \\
        \vdots    \\
        E (r_M^\text{i}, \theta_{M}^\text{i}, \phi_{M}^\text{i})
      \end{bmatrix}.
    \end{aligned}
  \end{equation}
  \medskip
  \hrule
\end{figure*}

The input-output behavior expression can be simplified under specific scenarios, particularly when the RIS comprises isotropic units. In isotropic scattering, incident EM waves uniformly scatter in all directions over the hemisphere of reflection, as depicted in Fig.~\ref{Isotropic}, irrespective of the angle of incidence or observation. The hypothetical isotropic element is sometimes referred to as Lambertian scatterers~\cite{maitre2013processing}. Consequently, the unit's scattering pattern simplifies to $\tau_n ( \theta^\text{s}, \phi^\text{s}; \theta^\text{i}, \phi^\text{i} ) = \tau$. The aggregated electric field is given as
\begin{equation}\label{Model_MISO2}
  \begin{aligned}
    S ( r^\text{s}, \theta^\text{s}, \phi^\text{s} ) = 
    & \sum_{m=1}^{M} E ( r^\text{i}_m, \theta^\text{i}_m, \phi^\text{i}_m )  \\
    & \sum_{n=1}^{N} \tau  e^{j \Omega_n} \frac{ e^{ - j 2 \pi r^\text{i}_m (n) / \lambda} e^{ - j 2 \pi r^\text{s} (n) / \lambda } }{ r^\text{i}_m (n)  r^\text{s} (n) }   . 
  \end{aligned}
\end{equation}

\begin{figure}[!htbp]
  \centering
  \includegraphics[width=0.6\linewidth]{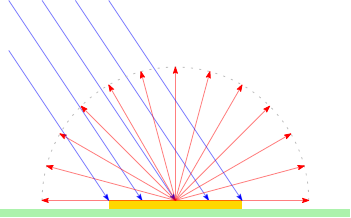}
  \caption{Lambertian scattering~\cite{maitre2013processing}. Incident EM waves uniformly scatter in all directions over the hemisphere of reflection.}
  \label{Isotropic}
\end{figure}

\begin{rem}
  In \eqref{Model_MISO} and \eqref{Model_MISO2}, we illustrate the linearization approximations inherent in the functionality of the RIS. Notably, we do not impose any specific topology of RISs or far-field regime restrictions during the derivation, which allows for a more general description of their behavior. Understanding the linear aggregation of forward incident signals offers valuable insights into the capabilities of RISs, laying a foundation for deploying multiple RISs in sensing applications.
\end{rem}

In the far-field regime, the distances $r^\text{s} (n)$ and $r^\text{i} (n)$ exhibit linear relationships with $r^\text{s}$ and $r^\text{i}$, respectively. This enables concise representations of the linear aggregations in \eqref{Model_MISO} and \eqref{Model_MISO2}. Following the approach outlined in~\cite{mi2023towards}, we can easily compute the path length differences $r^\text{s} (n) - r^\text{s}$ and $r^\text{i} (n) - r^\text{i}$. When analyzing uniform planar arrays, which represent the fundamental topology of RISs, we can employ the canonical linear representation \eqref{E:MIMO} to describe the input/output behaviors. In this context, we define $\mathbf{u} ( \theta, \phi) = [\sin \theta \cos \phi, \ \sin \theta \sin \phi, \ \cos \theta]^\top$.

To shorten notations, we let
\begin{align*}
    l (r^\text{s}) & = \frac{ e^{-j 2 \pi r^\text{s} / \lambda }}{ r^\text{s} } , \\
    \mathbf{E} ( \mathbf{r}^\text{i}, \mathbf{\Theta}^\text{i}, \mathbf{\Phi}^\text{i} ) & = \left[ E (r_1^\text{i}, \theta_{1}^\text{i}, \phi_{1}^\text{i}), \cdots, E (r_M^\text{i}, \theta_{M}^\text{i}, \phi_{M}^\text{i}) \right]^\top, \\
    l (\mathbf{r}^\text{i}) & = \left[ \frac{ e^{-j 2 \pi r_1^\text{i} / \lambda }}{ r_1^\text{i} }, \cdots, \frac{ e^{-j 2 \pi r_M^\text{i} / \lambda }}{ r_M^\text{i} } \right] , \\
    e^{j \mathbf{\Omega} } & = \left[  e^{j \Omega_1}, \cdots,  e^{j \Omega_N} \right] , \\
    \mathbf{v} ( \theta^\text{s}, \phi^\text{s} ) & = \left[ e^{ j 2 \pi \mathbf{p}_1^{\top} \mathbf{u} (\theta^\text{s}, \phi^\text{s})  / \lambda }, \cdots , e^{ j 2 \pi \mathbf{p}_N^\top \mathbf{u} (\theta^\text{s}, \phi^\text{s})  / \lambda } \right] , \\
    \mathbf{V} ( \mathbf{\Theta}^\text{i}, \mathbf{\Phi}^\text{i} ) & = 
    \begin{bmatrix}
      e^{ j 2 \pi \mathbf{p}_1^\top \mathbf{u} (\theta_1^\text{i}, \phi_1^\text{i})  / \lambda }  & \cdots  & e^{ j 2 \pi \mathbf{p}_1^\top \mathbf{u} (\theta_M^\text{i}, \phi_M^\text{i})  / \lambda } \\
      \vdots  & \ddots & \vdots \\
      e^{ j 2 \pi \mathbf{p}_N^\top \mathbf{u} (\theta_1^{\text{i}}, \phi_1^\text{i})  / \lambda }  & \cdots  & e^{ j 2 \pi \mathbf{p}_N^\top \mathbf{u} (\theta_M^\text{i}, \phi_M^\text{i})  / \lambda }\\
    \end{bmatrix} .
\end{align*}
Note that $\mathbf{v} ( \theta^\text{s}, \phi^\text{s} ) \in \mathbb{C}^{1 \times N}$ and $\mathbf{V} ( \mathbf{\Theta}^\text{i}, \mathbf{\Phi}^\text{i} ) \in \mathbb{C}^{N \times M}$.
Then, we represent the aggregated field at the observation point as
\begin{multline}\label{MISO_3}
  S (r^\text{s}, \theta^\text{s}, \phi^\text{s} ) \approx \overbrace{ \tau l (r^\text{s}) \mathbf{v} ( \theta^\text{s}, \phi^\text{s} ) \text{diag} (e^{j \mathbf{\Omega} }) \mathbf{V} ( \mathbf{\Theta}^\text{i}, \mathbf{\Phi}^\text{i} ) }^{ \mathbf{h}^\top_ {\mathbf{\Omega} } } \\
  \underbrace{ \text{diag} ( l (\mathbf{r}^\text{i}) ) \mathbf{E} ( \mathbf{r}^\text{i}, \mathbf{\Theta}^\text{i}, \mathbf{\Phi}^\text{i} ) }_{ \mathbf{E} ( \mathbf{\Theta}^\text{i}, \mathbf{\Phi}^\text{i} ) } . 
\end{multline}

To simplify notation, we denote the measured field as
\begin{equation}\label{MISO_4}
  S (r^\text{s}, \theta^\text{s}, \phi^\text{s} ) = \mathbf{h}^\top_ { \mathbf{\Omega} } \mathbf{E} ( \mathbf{\Theta}^\text{i}, \mathbf{\Phi}^\text{i} ) + \nu . 
\end{equation}
Here, $\nu$ represents the noise, which includes modeling inaccuracies resulting from approximations in the formulation as well as other inherent uncertainties in the measurements. Without loss of generality, we assume the noise follows a distribution with a zero mean and variance $\sigma^2$.

\subsection{Configurational Diversity}
One significant advantage of the RIS-centric backward sensing approach is its configurational diversity, which enables the generation of multiple measurement vectors (several snapshots). We can rewrite the measurement of the aggregated field corresponding to various configurations as 
\begin{multline*}
  S_t (r^\text{s}, \theta^\text{s}, \phi^\text{s} ) = \tau l (r^\text{s}) e^{j \mathbf{\Omega}_t }  \text{diag} ( \mathbf{v} ( \theta^\text{s}, \phi^\text{s} ) ) \\
 \mathbf{V} ( \mathbf{\Theta}^\text{i}, \mathbf{\Phi}^\text{i} ) \mathbf{E} ( \mathbf{\Theta}^\text{i}, \mathbf{\Phi}^\text{i} ) + \nu_t, \ t = 1, \ldots, T . 
\end{multline*}
The $T$ snapshots due to various configurations can be stacked into a column vector. Thus we have
\begin{multline}\label{E:MultiShot}
  \mathbf{S} (r^\text{s}, \theta^\text{s}, \phi^\text{s} ) = \overbrace{ \tau l (r^\text{s}) \mathbf{e}^{j \mathbf{\Omega} } \text{diag} \left( \mathbf{v} ( \theta^\text{s}, \phi^\text{s} ) \right) \mathbf{V} ( \mathbf{\Theta}^\text{i}, \mathbf{\Phi}^\text{i} ) }^{ \mathbf{H}( \mathbf{\Omega} ) } \\
  \text{diag} ( l (\mathbf{r}^\text{i}) ) \mathbf{E} ( \mathbf{r}^\text{i}, \mathbf{\Theta}^\text{i}, \mathbf{\Phi}^\text{i} ) + \mathbf{n} .
\end{multline}
Here $\mathbf{S} (r^\text{s}, \theta^\text{s}, \phi^\text{s} ) = \left[ S_1 (r^\text{s}, \theta^\text{s}, \phi^\text{s} ), \cdots, S_T (r^\text{s}, \theta^\text{s}, \phi^\text{s} ) \right]^\top$, $\mathbf{e}^{j \mathbf{\Omega} } = \left[ e^{j \mathbf{\Omega}_1 } ; \cdots; e^{j \mathbf{\Omega}_T }  \right]$ and $\mathbf{n} = \left[ \nu_1, \cdots, \nu_T \right]^\top$.

Denoting the measurement operator configured by $\mathbf{\Omega}$ as 
\begin{equation}\label{E:H_Omega}
  \mathbf{H} ( \mathbf{\Omega} ) = \tau l (r^\text{s}) \mathbf{e}^{j \mathbf{\Omega} } \text{diag} \left( \mathbf{v} ( \theta^\text{s}, \phi^\text{s} ) \right) \mathbf{V} ( \mathbf{\Theta}^\text{i}, \mathbf{\Phi}^\text{i} ) ,
\end{equation}
we obtain the canonical representation for the input/output behaviors corresponding to a single RIS 
\begin{equation}\label{E:MultiShot2}
  \mathbf{S} (r^\text{s}, \theta^\text{s}, \phi^\text{s} ) = \mathbf{H} ( \mathbf{\Omega} ) \mathbf{E} ( \mathbf{\Theta}^\text{i}, \mathbf{\Phi}^\text{i} ) + \mathbf{n}.
\end{equation}

\begin{rem}
One advantage of the canonical representation \eqref{E:MultiShot2} is its use of a simple system of linear equations to describe the forward signal aggregation process via RISs under reasonable assumptions. This system is both physically accurate and mathematically tractable, serving as a fundamental module in the RIS-centric backward sensing approach. This canonical system facilitates the design of backward sensing algorithms and enables comprehensive performance analyses.
\end{rem}

\begin{rem}
The matrix $\mathbf{V} ( \mathbf{\Theta}^\text{i}, \mathbf{\Phi}^\text{i} )$ represents the multiplexing of waves from the sources to each element of the RIS. A characteristic of this matrix is that if two sources have coherent positions, such as being aligned along the same line when observed via a uniform linear RIS, $\mathbf{V} ( \mathbf{\Theta}^\text{i}, \mathbf{\Phi}^\text{i} )$ can become rank-deficient, leading to the sensing matrix $ \mathbf{H} ( \mathbf{\Omega} )$ becoming ill-conditioned. This deficiency results in degraded performance in the backward sensing approach. When deploying RIS, we should carefully consider the geometry and topology to minimize the risk of ill-conditioning in $\mathbf{V} ( \mathbf{\Theta}^\text{i}, \mathbf{\Phi}^\text{i} )$.
\end{rem}

\section{Backward Sensing Using Multiple RISs} \label{S:BackwardSensing}

It is important to note that while incident distances are accounted for in the canonical representations \eqref{MISO_3} and \eqref{E:MultiShot}, $\mathbf{E} ( \mathbf{\Theta}^\text{i}, \mathbf{\Phi}^\text{i} ) = \text{diag} ( l (\mathbf{r}^\text{i}) ) \mathbf{E} ( \mathbf{r}^\text{i}, \mathbf{\Theta}^\text{i}, \mathbf{\Phi}^\text{i} )$ actually represents the incident fields arriving at the RIS along $( \mathbf{\Theta}^\text{i}, \mathbf{\Phi}^\text{i} )$. In other words, $l(\mathbf{r}^\text{i})$ serves as a distance normalization term. The role of $l (\mathbf{r}^\text{i})$ becomes essential in localizations involving multiple RISs, where large-scale distance attenuation and phase delay factors can be detected by different RISs. Before exploring the capabilities provided by deploying multiple RISs, we first clarify the DoA estimation function using a single RIS.

\subsection{DoA Estimation via a Single RIS}

For the DoA estimation task using a single RIS, the vector $\mathbf{E} ( \mathbf{\Theta}^\text{i}, \mathbf{\Phi}^\text{i} )$ in the canonical representation \eqref{E:MultiShot2} is unknown and needs to be estimated. A fundamental assumption of this scheme is that the sensing operator $\mathbf{H} ( \mathbf{\Omega} )$ is known a priori, which necessitates knowledge of $l (r^\text{s}) \text{diag} \left( \mathbf{v} ( \theta^\text{s}, \phi^\text{s} ) \right)$, $\mathbf{e}^{j \mathbf{\Omega} }$, and $\mathbf{V} ( \mathbf{\Theta}^\text{i}, \mathbf{\Phi}^\text{i} )$. The term $l (r^\text{s}) \text{diag} \left( \mathbf{v} ( \theta^\text{s}, \phi^\text{s} ) \right)$ is relevant to the receiver's location, while $\mathbf{e}^{j \mathbf{\Omega} }$ is determined by the configurations. Both terms are available within the backward sensing scenario.

The situation becomes more complicated for $\mathbf{V} ( \mathbf{\Theta}^\text{i}, \mathbf{\Phi}^\text{i} )$. This matrix is governed by the directions of the incident waves, the central task of DoA estimation. One potential approach involves exploring densely sampled grids across the RoI. While this manipulation significantly expands the solution space, the extensive reconfiguration capabilities ensure the problem remains well-defined under modest assumptions (provided the sampling points are not too closely spaced). This approach is illustrated in Fig.~\ref{F:SamplingOneAntenna} with a linear RIS for better understanding. The authors have previously investigated the forward signal aggregations using linear RISs in \cite{mi2023towards,xiong2024optimal}.

\begin{figure}[!htbp] 
  \centerline{\includegraphics[width=0.9\columnwidth]{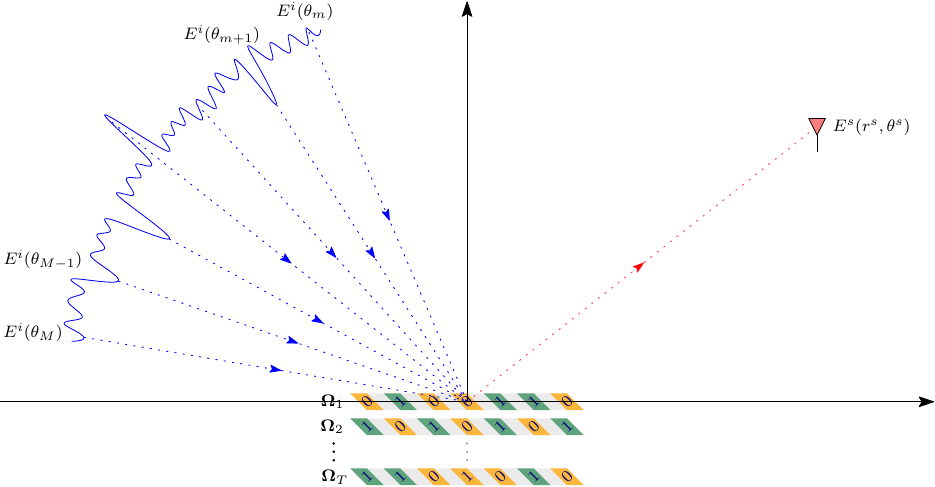}} 
  \caption{Uniform discretization across the azimuth angle domain (RoI). A total of $T$ snapshots are captured corresponding to various configurations.}
  \label{F:SamplingOneAntenna}
\end{figure}

When the sensing operator $\mathbf{H} ( \mathbf{\Omega} )$ is given, ensuring that a set of measurements $\mathbf{S} ( r^\text{s}, \theta^\text{s}, \phi^\text{s} )$ corresponds uniquely to $\mathbf{E} ( \mathbf{\Theta}^\text{i}, \mathbf{\Phi}^\text{i} )$ requires that the rank of $\mathbf{H} ( \mathbf{\Omega} )$ be at least equal to the dimension of RoI. In a typical scenario where an RIS is configured to produce a well-conditioned full column rank matrix $\mathbf{H} ( \mathbf{\Omega} )$, according to \eqref{E:MultiShot2}, a reliable estimation for the incident field can be given by 
\begin{equation}\label{E:InverseSensing}
  \widehat{\mathbf{E}} ( \mathbf{\Theta}^\text{i}, \mathbf{\Phi}^\text{i} ) = \bigl( \mathbf{H} ( \mathbf{\Omega} )^* \mathbf{H} ( \mathbf{\Omega} ) \bigr)^{-1} \mathbf{H} ( \mathbf{\Omega} )^* \mathbf{S} ( r^\text{s}, \theta^\text{s}, \phi^\text{s} ) .
\end{equation}

\begin{rem}
Both \eqref{E:MultiShot2} and \eqref{E:InverseSensing} indicate that a single RIS is limited to determining the directions of incident waves, making tasks such as location or imaging infeasible. For accurate localization, multiple strategically positioned RISs are required. 
\end{rem}

\subsection{Localization Using Multiple RISs}

Similar to the scenario of DoA estimation, where incident angles are discretized, the localization using multiple RISs with spatial diversity also requires volume discretization. A common method involves direct discretization in Cartesian coordinates, as illustrated in Fig.~\ref{F:TwoModes}. This technique typically employs uniform voxels, allowing for straightforward implementation. The advantage of this strategy lies in its versatility and generalizability. However, such discretization may lead to ill-conditioning in the sensing operator, similar to the situation where two sources align along the same line during sensing with uniform linear RISs. We will explore this topic further in Section~\ref{S:Experiments}.

Unlike scenarios involving a single RIS, the operational mode of receivers is crucial in sensing with multiple RISs. The most straightforward mode is each RIS having its dedicated receiver, as depicted in Fig.~\ref{F:Mode1}. Alternatively, there is a mode where the sensing system uses a single receiver. In this setup, the reflected signals from all RISs are aggregated to the receiver, as illustrated in Fig.~\ref{F:Mode2}. Next, we will examine the linear representation of measurement operators in localization using multiple RISs. The form of this operator varies based on the operational mode of the receivers.

\begin{figure}[!htbp]
  \centering
  \subfigure[]{
  \label{F:Mode1}
  \includegraphics[width=.6\columnwidth]{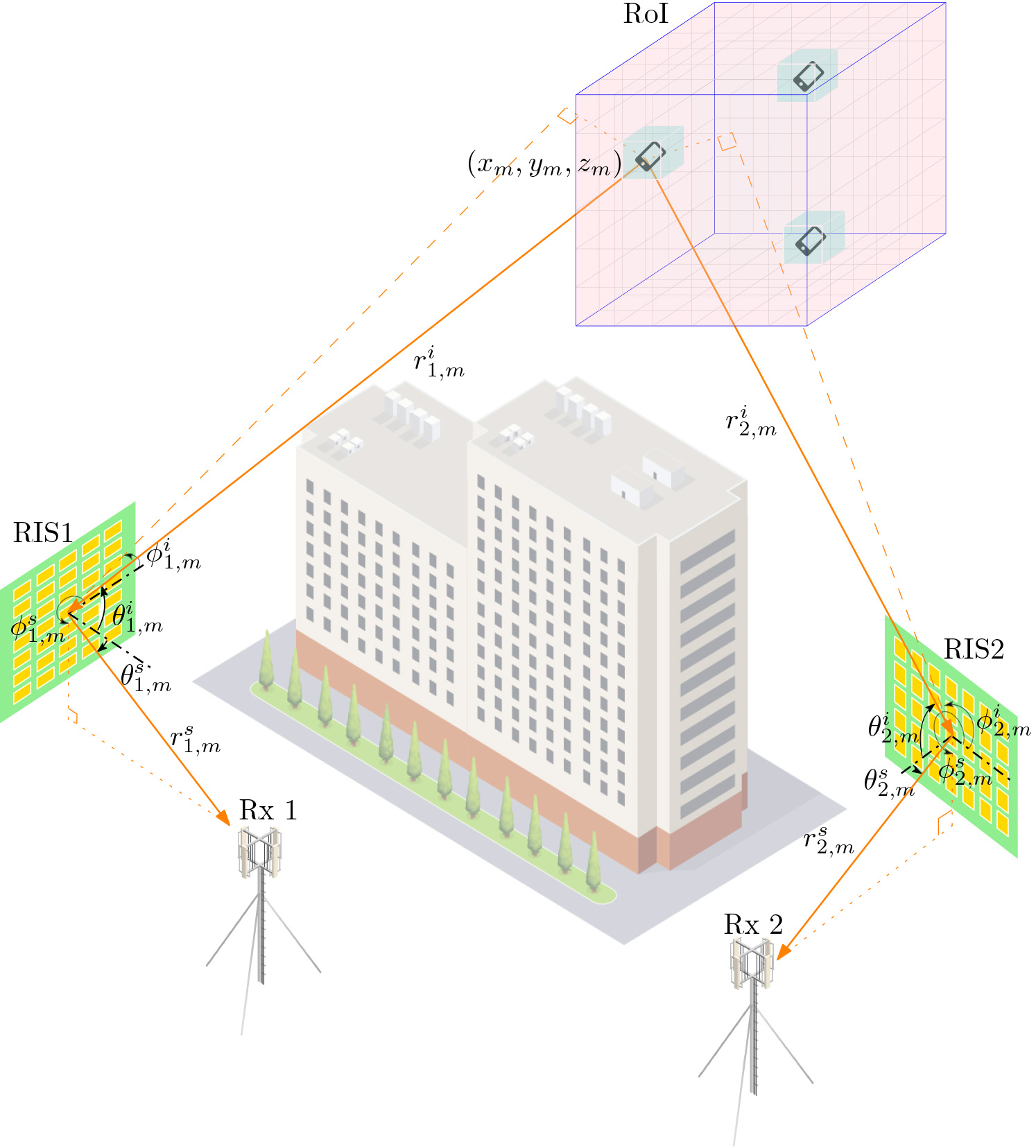}}
  \subfigure[]{
  \label{F:Mode2}
  \includegraphics[width=.6\columnwidth]{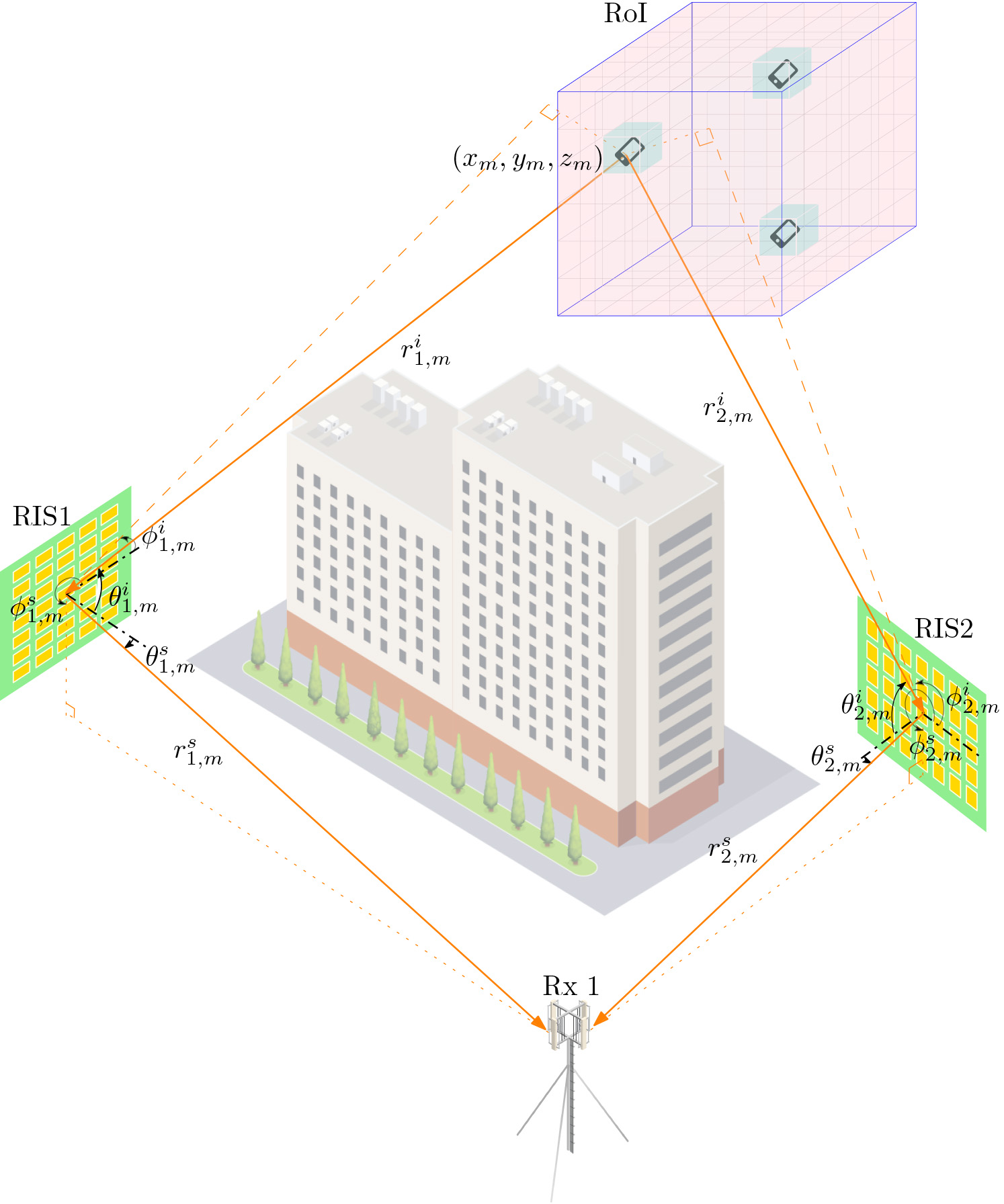}}
  \caption{ Two modes in the RIS-centric sensing system with direct discretization in Cartesian coordinates, where each voxel represents a defined volume: (a) Each RIS is equipped with a dedicated receiver. (b) The system employs a single receiver for aggregating reflected signals from all RISs.} 
  \label{F:TwoModes} 
\end{figure}

In the context of backward sensing, the geometric and topological relationships among the RoI, receivers, and RISs are predetermined. For convenience, we denote by $\mathbf{E} ( \mathbf{x}, \mathbf{y}, \mathbf{z} )$ the incident field in Cartesian coordinates and use the subscript $k$ to index RISs. Utilizing the canonical aggregation \eqref{E:MultiShot2}, we represent the input/output behavior corresponding to each individual RIS as
\[
  \mathbf{S} (r^\text{s}_k, \theta^\text{s}_k, \phi^\text{s}_k ) = \mathbf{H} ( \mathbf{\Omega}_k ) \text{diag} ( l (\mathbf{r}^\text{i}_k ) ) \mathbf{E} ( \mathbf{x}, \mathbf{y}, \mathbf{z} ) + \mathbf{n}_k.
\]

For the mode where each RIS has its dedicated receiver, the measurement vectors corresponding to each RIS can be stacked into a column vector. Thus, we can represent the measurement process associated to the whole system as
\begin{equation}\label{E:Mode_1}
  \begin{bmatrix}
    \mathbf{S} (r^\text{s}_1, \theta^\text{s}_1, \phi^\text{s}_1 ) \\
    \vdots \\
    \mathbf{S} (r^\text{s}_K, \theta^\text{s}_K, \phi^\text{s}_K )
  \end{bmatrix}
  =
  \begin{bmatrix}
    \mathbf{H} ( \mathbf{\Omega}_1 ) \text{diag} ( l (\mathbf{r}^\text{i}_1 ) ) \\
    \vdots \\
    \mathbf{H} ( \mathbf{\Omega}_K ) \text{diag} ( l (\mathbf{r}^\text{i}_K ) )
  \end{bmatrix} 
  \mathbf{E} ( \mathbf{x}, \mathbf{y}, \mathbf{z} ) + \mathbf{n}.
\end{equation}
On the other hand, when the sensing system comprises only one receiver that serves all RISs, we can represent the measurement process
\begin{equation}\label{E:Mode_2}
\begin{aligned}
\mathbf{S} ( r^\text{s}, \theta^\text{s}, \phi^\text{s} ) 
= & \mathbf{S} (r^\text{s}_1, \theta^\text{s}_1, \phi^\text{s}_1 ) + \cdots + \mathbf{S} (r^\text{s}_K, \theta^\text{s}_K, \phi^\text{s}_K ) \\
= &  \bigl[ \mathbf{H} ( \mathbf{\Omega}_1 ) \text{diag} ( l (\mathbf{r}^\text{i}_1 ) ) + \cdots \\
& + \mathbf{H} ( \mathbf{\Omega}_K ) \text{diag} ( l (\mathbf{r}^\text{i}_K ) ) \bigr] \mathbf{E} ( \mathbf{x}, \mathbf{y}, \mathbf{z} ) + \mathbf{n}.
\end{aligned}
\end{equation}

\begin{rem}
Both \eqref{E:Mode_1} and \eqref{E:Mode_2} are systems of linear equations describing forward signal aggregation. The distinction between these systems arises from whether each RIS has its dedicated receiver. In scenarios with specific geometries and topologies, these systems may exhibit distinct structures.
\end{rem}

For simplicity, we represent the measurement process associated with multiple RISs as
\begin{equation}\label{E:LinearMeasurement}
  \mathbf{S} = \mathbf{H} \mathbf{E} + \mathbf{n} .
\end{equation}
The localization problem can be formulated as follows: given the observation $\mathbf{S}$ and the sensing operator $\mathbf{H}$, find a suitable RoI $\mathbf{E}$ that satisfies \eqref{E:LinearMeasurement}.

Similar to DoA estimation via a single RIS, if the length of the observation vector $\mathbf{S}$ is greater than the degree of RoI $\mathbf{E}$, and the sensing operator $\mathbf{H}$ is of full column rank and well-conditioned, the most straightforward estimation of the RoI is the least squares solution to \eqref{E:LinearMeasurement}, which is expressed as
\begin{equation}\label{E:InverseSensingMultiple}
  \hat {\mathbf{E}} = \mathbf{H}^\dag \mathbf{S} = \bigl( \mathbf{H}^* \mathbf{H}  \bigr)^{-1} \mathbf{H}^* \mathbf{S} .
\end{equation}

\begin{rem}
In the proposed RIS-centric backward sensing approach, we model the measurement process in the far-field regime using a system of linear equations. Various reconstruction algorithms can provide high-fidelity solutions, but exploring these algorithms exceedes the scope of this paper. The advantages of the least squares approach lie in its simplicity, broad applicability, and strong theoretical foundation. There exist well-established tools for analyzing the properties of least squares solutions. 
\end{rem}

\section{Key Performance Indicators}\label{S:KeyIndicators}  

This section is to identify the key indicators that dominate the performance of the proposed RIS-centric sensing approaches. It's important to note that the indicators affecting sensing performance are algorithm-dependent. This paper focuses on utilizing least squares solutions, as illustrated in \eqref{E:InverseSensing} and \eqref{E:InverseSensingMultiple}. If the measurement/sensing operator is a well-conditioned full column rank matrix, the least squares solution provides a reliable estimation for the incident field. Therefore, we examine sensing performance from both the rank and conditioning perspectives.

\subsection{Rank of the Sensing Operators}

Suppose the sensing system comprises $K$ RISs, with each RIS composed of $N_k$ elements. These RISs are collaboratively configured to sense an RoI with dimension $M$. In the mode where each RIS has its dedicated receiver, the number of snapshots for each RIS is denoted by $T_k$. On the other hand, in the setup where the sensing system features only one receiver, the number of snapshots is denoted by $T$. The following theorem clarifies the rank of the sensing operators.

\begin{thm}\label{T:1}
  In the mode where each RIS has its dedicated receiver, the rank of the sensing matrix is bounded above by the dimension of the RoI, the total number of measurements, and the total number of elements. Specifically, we have
  \[
    \normalfont 
    \text{rank}(\mathbf{H}) \le \min \left \{ M, \sum_{k=1}^K \min \{ T_k, N_k \} \right \}.
  \]
  Furthermore, in the scenario where the sensing system utilizes only one receiver, the rank of the sensing matrix is bounded by
  \[
    \normalfont 
    \text{rank}(\mathbf{H}) \le \min \left \{ M, T,  \sum_{k=1}^K N_k \right \}.
  \]
\end{thm}

\begin{proof}
We first establish the proof for the mode where each RIS has its dedicated receiver. Since $\mathbf{H}$ consists of $M$ columns, $\text{rank}(\mathbf{H}) \le M$. For each $k$, we note that
\[
\text{rank} \left ( \mathbf{H} ( \mathbf{\Omega}_k ) \text{diag} ( l (\mathbf{r}^\text{i}_k ) ) \right) = \text{rank} \left ( \mathbf{H} ( \mathbf{\Omega}_k )  \right) \le \min \{ T_k, N_k \} .
\]
The last inequality arises from the fact that $\mathbf{H} ( \mathbf{\Omega}_k )$ is a product of three matrices, as illustrated in \eqref{E:MultiShot}. The rank of the first matrix $\mathbf{e}^{j \mathbf{\Omega}_k }$ is bounded by $T_k$, and the rank of the second matrix is $N_k$. Because $\mathbf{H}$ is the stacking of $K$ such matrices, we have
\[
  \text{rank}(\mathbf{H}) \le \sum_{k=1}^K \min \{ T_k, N_k \} .
\]
The proof is similar for the mode where the sensing system utilizes only one receiver. 
\end{proof}

\begin{rem}
In an RIS-centric backward sensing scheme, the rank of the sensing matrix indicates how many degrees of freedom of the scene can be measured or coded by this operator. Theorem~\ref{T:1} specifies that, to achieve high-fidelity recovery of the RoI, a necessary condition is that the dimension of the RoI is less than both the total number of elements and the number of measurements. This finding has significant implication for the understanding that when the size of RISs is predetermined, the dimension of RoI is bounded by this size, even if the number of measurements can be greatly increased. 
\end{rem}

\subsection{The Spatial Resolution Analysis}

We now investigate the resolution associated with least squares solutions. For clarity, we focus solely on exploring the angular resolution corresponding to the uniform linear RIS. In fact, the relative error of a least squares solution is given by
\begin{equation}\label{E:RelativeError}
  \begin{aligned}
  \frac{ \lVert \hat {\mathbf{E}} - \mathbf{E} \rVert } { \lVert \mathbf{E} \rVert } 
  = & \frac{ \lVert {\mathbf{H}}^{\dag} \mathbf{H} \mathbf{E} - \mathbf{E} + {\mathbf{H}}^{\dag} \mathbf{n} \rVert } { \lVert \mathbf{E} \rVert } \\
  \le & \frac{ \lVert {\mathbf{H}}^{\dag} \mathbf{H} \mathbf{E} - \mathbf{E} \rVert + \lVert {\mathbf{H}}^{\dag} \mathbf{n} \rVert } { \lVert \mathbf{E} \rVert } \\
  \le & \lVert {\mathbf{H}}^{\dag} \mathbf{H} - \mathbf{I} \rVert + \lVert {\mathbf{H}}^{\dag} \rVert \lVert \mathbf{n} \rVert / \lVert \mathbf{E} \rVert \\
  = & \lVert {\mathbf{H}}^{\dag} \mathbf{H} - \mathbf{I} \rVert + \sqrt{T} \lVert {\mathbf{H}}^{\dag} \rVert / \sqrt{\mathrm{SNR}} .
  \end{aligned}
\end{equation}
For simplicity, we define $\mathrm{SNR} = \lVert \mathbf{E} \rVert^2 / \sigma^2$.

The relative error is bounded above by the summation of two terms. When the measurement operator $\mathbf{H}$ is a full column rank matrix, 
achieving $\lVert {\mathbf{H}}^{\dag} \mathbf{H} - \mathbf{I} \rVert \approx 0$ is feasible. Furthermore, if $\lVert {\mathbf{H}}^{\dag} \rVert$ 
remains bounded by a reasonable value, then recovery using the least squares solution ensures high fidelity. It should be noted that $\lVert {\mathbf{H}}^{\dag} \rVert = 1 / \sigma_{\text{min}} ( \mathbf{H} )$. We will then examine the minimum singular value of $\mathbf{H}$.

In a sensing system, the minimum resolvable distance is defined as the spatial resolution. The underlying motivation of this definition is that when two point sources are in close proximity, the system cannot distinguish them as two distinct points. Consider a scenario where only two incident waves, one along $\theta^{\text{i}}$ and another along $\theta^{\text{i}}+\Delta$, are impinging on the RIS. According to~\eqref{E:H_Omega}, the measurement operator is reduced to 
\[
  \mathbf{H} ( \theta^{\text{i}}, \theta^{\text{i}}+\Delta ) = \tau l (r^\text{s}) \mathbf{e}^{j \mathbf{\Omega} } \text{diag} \left( \mathbf{v} ( \theta^\text{s} ) \right) \mathbf{V} ( \theta^{\text{i}}, \theta^{\text{i}}+\Delta ) .
\]

When $\Delta \neq 0$, the matrix $\mathbf{V} ( \theta^{\text{i}}, \theta^{\text{i}}+\Delta )$ is of full column rank. And the matrix $\text{diag} \left( \mathbf{v} ( \theta^\text{s} ) \right)$ is also of full rank. Therefore, we have
\begin{equation}\label{E:SingularValueBound}
  \begin{aligned}
    & \sigma_{\text{min}} ( \mathbf{H} (\theta^{\text{i}}, \theta^{\text{i}}+\Delta) ) \\
    \ge & \lvert \tau l (r^\text{s}) \rvert \sigma_{\text{min}} ( \mathbf{e}^{j \mathbf{\Omega} } )  \sigma_{\text{min}} \left( \text{diag} ( \mathbf{v} ( \theta^\text{s} ) ) \right) \sigma_{\text{min}} \left( \mathbf{V} ( \theta^{\text{i}}, \theta^{\text{i}}+\Delta ) \right) \\
    = & \lvert \tau l (r^\text{s}) \rvert \sigma_{\text{min}} ( \mathbf{e}^{j \mathbf{\Omega} } ) \sigma_{\text{min}} \left( \mathbf{V} ( \theta^{\text{i}}, \theta^{\text{i}}+\Delta ) \right) .
  \end{aligned}
\end{equation}
This inequality is a consequence of the lower bound for the smallest singular value of a product of two matrices, as illustrated in Lemma~\ref{L:3}.

\begin{lem}\label{L:3}
If $\mathbf{B}$ has full column rank, then $\sigma_{\text{min}} ( \mathbf{A} \mathbf{B} ) \ge \sigma_{\text{min}} ( \mathbf{A} ) \sigma_{\text{min}} ( \mathbf{B} )$.
\end{lem}

\begin{proof}
Given that $\mathbf{B}$ has full column rank, for any $x \neq 0$, we have $\mathbf{B} x \neq 0$. Then,
  \begin{multline*}
    \sigma_{\text{min}} ( \mathbf{A} \mathbf{B} ) 
    = \min_{x \neq 0} \frac{ \lVert \mathbf{A} \mathbf{B} x \rVert }{ \lVert x \rVert } 
    = \min_{x \neq 0} \frac{ \lVert \mathbf{A} \mathbf{B} x \rVert }{ \lVert \mathbf{B} x \rVert } \frac{ \lVert \mathbf{B} x \rVert }{ \lVert x \rVert } \\
    \ge \min_{y \neq 0} \frac{ \lVert \mathbf{A} y \rVert }{ \lVert y \rVert } \min_{x \neq 0} \frac{ \lVert \mathbf{B} x \rVert }{ \lVert x \rVert }
    = \sigma_{\text{min}} ( \mathbf{A} ) \sigma_{\text{min}} ( \mathbf{B} ) .
  \end{multline*}
\end{proof}

\begin{lem}\label{L:4}
Consider a Vandermonde matrix given by
\begin{equation}\label{E:VandermodeMatrix}
  \mathbf{V} ( \theta^{\mathrm{i}}, \theta^{\mathrm{i}}+\Delta ) =
  \begin{bmatrix}
    1                                                   &  1                                                   \\
    e^{ j 2 \pi d \sin \theta^{\mathrm{i}} / \lambda }             &  e^{ j 2 \pi d \sin ( \theta^{\mathrm{i}}+\Delta ) / \lambda }             \\
    \vdots                                              &  \vdots                                              \\
    e^{ j 2 \pi (N-1) d \sin \theta^{\mathrm{i}} / \lambda } &  e^{ j 2 \pi (N-1) d \sin ( \theta^{\mathrm{i}}+\Delta ) / \lambda }
\end{bmatrix} ,
\end{equation}
where $\Delta \neq 0$. The singular values of $\mathbf{V} ( \theta^{\mathrm{i}}, \theta^{\mathrm{i}}+\Delta )$ are given by
\begin{multline*}
  \sigma_{\text{max}} ( \mathbf{V} ( \theta^{\mathrm{i}}, \theta^{\mathrm{i}}+\Delta ) ) \\
  = \sqrt{ N + \left| \frac{ \sin ( \pi N d (\sin \theta^{\mathrm{i}} - \sin ( \theta^{\mathrm{i}}+\Delta ) ) / \lambda ) }{ \sin ( \pi d ( \sin \theta^{\mathrm{i}} - \sin ( \theta^{\mathrm{i}}+\Delta ) ) / \lambda ) } \right|  } ,
\end{multline*}
and 
\begin{multline}\label{E:LeastSingularValue1}
  \sigma_{\text{min}} ( \mathbf{V} ( \theta^{\mathrm{i}}, \theta^{\mathrm{i}}+\Delta ) ) \\
  = \sqrt{ N - \left| \frac{ \sin ( \pi N d (\sin \theta^{\mathrm{i}} - \sin ( \theta^{\mathrm{i}}+\Delta ) ) / \lambda ) }{ \sin ( \pi d ( \sin \theta^{\mathrm{i}} - \sin ( \theta^{\mathrm{i}}+\Delta ) ) / \lambda ) } \right|  } .
\end{multline}
\end{lem}

\begin{proof}
See the Appendix~\ref{S:Appendix_A3}
\end{proof}

\begin{lem}\label{L:5}
For $x$ sufficiently small, we have
\begin{equation}\label{E:sin_sin}
  \frac{ \sin (N x) }{\sin (x)} \approx N - \frac{ N ( N^2 - 1 ) x^2 }{ 6 } .
\end{equation}
\end{lem}

\begin{proof}
  See the Appendix~\ref{S:Appendix_A4}
\end{proof}

Combining \eqref{E:LeastSingularValue1} and \eqref{E:sin_sin} yields an approximation for the smallest singular value of $\mathbf{V}(\theta^{\text{i}}, \theta^{\text{i}}+\Delta)$
\begin{equation}\label{E:ApproximationV}
\begin{aligned}
  & \sigma_{\text{min}} ( \mathbf{V} ( \theta^{\text{i}}, \theta^{\text{i}}+\Delta ) ) \\
  \approx & \frac{ \pi }{ \sqrt{6} } \frac{ d }{ \lambda } \sqrt{ N ( N^2 - 1 ) } \lvert \sin \theta^{\text{i}} - \sin ( \theta^{\text{i}} + \Delta ) \rvert  \\
  \approx & \frac{ \pi }{ \sqrt{6} } \frac{ d }{ \lambda } \sqrt{ N ( N^2 - 1 ) } \lvert \Delta \rvert \cos \theta^{\text{i}} .
\end{aligned}
\end{equation}
In the last step, we utilize the first-order Taylor expansion of $\sin ( \theta^{\text{i}} + \Delta )$ at $\theta^{\text{i}}$.

Next, we will examine $\sigma_{\text{min}} ( \mathbf{e}^{j \mathbf{\Omega} } )$ using the Marchenko-Pastur law~\cite{marchenko1967distribution}, which characterizes the asymptotic behavior of singular values of large rectangular random matrices. Let's consider a random matrix $\mathbf{X}$ of size $T \times N$ $(T>N)$, whose entries are independent identically distributed random variables with mean 0 and variance 1. We define the empirical spectral distribution of $\mathbf{Y}_N = \mathbf{X}^* \mathbf{X} / T$ as 
\[
\mu_{ \mathbf{Y}_N } \equiv \frac{1}{N} \sum_{i=1}^{N} \delta_{\lambda_i( \mathbf{Y}_N )},
\]
where $\lambda_1( \mathbf{Y}_N ) \le \cdots \le \lambda_N( \mathbf{Y}_N )$ are the eigenvalues (including multiplicity) and $\delta_{ \lambda_i( \mathbf{Y}_N ) }(x)$ is the indicator function $\mathbf{1}_{ \lambda_i( \mathbf{Y}_N ) \le x }$.

A celebrated theorem in random matrix theory, known as the Marchenko-Pastur distribution, describes the limiting spectral distribution of  $\mathbf{Y}_N$ \cite{marchenko1967distribution, anderson2010introduction}. Under the Kolmogorov condition, i.e., $T, N \to \infty$ and $N / T \to c \in (0, 1)$, $\mu_{ \mathbf{Y}_N } $ converges weakly and almost surely to $\mu_{\text{M-P}}$ with a density function
\[
d \mu_{\text{M-P}} (x) = \frac{1}{2 \pi c x} \sqrt{(x-a)^{+}(b-x)^{+}} d x .
\]
Here $a = (1 - \sqrt{c})^2$, $b = (1 + \sqrt{c})^2$, and $a^{+} = \max (0, a)$. The Marchenko-Pastur distribution implies, as $T, N \to \infty$, the probability of the eigenvalues of $\mathbf{Y}_N$ lying outside the interval $[a, b]$ approaches 0.

We now turn our attention to the matrix $\mathbf{e}^{j \mathbf{\Omega}}$, whose entries are independent identically distributed random variables with mean 0 and variance 1. By leveraging the Marchenko-Pastur law, we can obtain a good approximation for its smallest singular value in the regime of large $T$ and $N$,
\begin{equation}\label{E:ApproximationE}
  \sigma_{\text{min}} ( \mathbf{e}^{j \mathbf{\Omega} } ) \approx \sqrt{T} \left( 1 - \sqrt{ \frac{N}{T}} \right) = \sqrt{T} - \sqrt{N}. 
\end{equation}

Substituting \eqref{E:ApproximationV} and \eqref{E:ApproximationE} into \eqref{E:SingularValueBound} yields
\begin{equation*}
\begin{aligned}
  & \sigma_{\text{min}} ( \mathbf{H} (\theta^{\text{i}}, \theta^{\text{i}}+\Delta) ) \\
  \gtrsim & \frac{ \pi }{ \sqrt{6} } \tau ( r^\text{s} )^{-1} \frac{ d }{ \lambda }  ( \sqrt{T} - \sqrt{N} ) \sqrt{ N ( N^2 - 1 ) } \lvert \Delta \rvert \cos \theta^{\text{i}}  \\
  \approx & \frac{ \pi }{ \sqrt{6} } \tau ( r^\text{s} )^{-1} \frac{ d }{ \lambda }  \left( 1 - \sqrt{ \frac{N}{T} } \right) T^{ \frac{1}{2} } N^{\frac{3}{2}} \lvert \Delta \rvert \cos \theta^{\text{i}}.
\end{aligned}
\end{equation*}
It then follows that
\begin{multline*}
  \left\lVert {\mathbf{H} (\theta^{\text{i}}, \theta^{\text{i}}+\Delta) }^+ \right\rVert \lesssim  \frac{\sqrt{6}}{\pi} r^\text{s} \tau^{-1} \left( \frac{ d }{ \lambda } \right)^{-1}  \left( 1 - \sqrt{ \frac{N}{T} } \right)^{-1} \\
  T^{ - \frac{1}{2} } N^{ - \frac{3}{2} } \lvert \Delta \rvert ^{-1} ( \cos \theta^{\text{i}} )^{-1}. 
\end{multline*}
We finally obtain the main theorem, as detailed below.

\begin{thm}\label{T:RelativeError}
If only two closely incident waves along $\theta^{\mathrm{i}}$ and $\theta^{\mathrm{i}}+\Delta$ are impinging on a uniform linear RIS, the relative error for the least squares solution is bounded approximately
\begin{multline}\label{E:RelativeErrorBound}
  \frac{ \lVert \hat {\mathbf{E}} - \mathbf{E} \rVert } { \lVert \mathbf{E} \rVert } \lesssim \frac{\sqrt{6}}{\pi} r^{\mathrm{s}} \tau^{-1} \left( \frac{ d }{ \lambda } \right)^{-1} \left( 1 - \sqrt{ \frac{N}{T} } \right)^{-1} \\
  N^{ - \frac{3}{2} } \lvert \Delta \rvert^{-1} ( \cos \theta^{\mathrm{i}} )^{-1} \mathrm{SNR}^{-1/2}.
\end{multline}
\end{thm}

\begin{rem}
Theorem~\ref{T:RelativeError} clarifies the factors that affect the sensing performance of a single uniform linear RIS. As seen from \eqref{E:RelativeErrorBound}, in addition to the SNR, three distinct categories of factors play significant roles in determining the relative error. The first category relates to the topology and geometry of the RIS, including parameters such as $\tau$ (the scattering pattern of the element), $N$ (the number of elements) and $d$ (the spacing between elements). By rewriting $( d / \lambda )^{-1} N^{ - \frac{3}{2} }$ as $( N d / \lambda )^{-1} N^{ - \frac{1}{2} }$, we observe that the bound is inversely proportional to $N d$, representing the aperture of the RIS. Unlike the situation in forward beamforming, where undesirable grating lobes arise when $d / \lambda > 1/2$, there is no such constraint on the spacing $d$ in backward sensing. Indeed, increasing $d$ can enhance the aperture, leading to greater gains. However, it's essential to note that increasing $N$ (the number of elements) can potentially provide more gains. This observation can be attributed to the fact that the power reflected by RISs grows quadratically with the total collecting area of the elements, as clarified in Section~\ref{S:Experiments}.
\end{rem}

\begin{rem}
The second category that affects sensing performance is related to incident angles and the receiver's position. As the receiver moves closer to the RIS, it captures more power, resulting in lower recovery errors. This relationship is detailed in Eq.~\eqref{E:RelativeErrorBound}, where the recovery error is directly proportional to $r^\text{s}$. Moreover, the relative error is inversely proportional to $\lvert \Delta \rvert$, providing insight into the spatial resolution performance. Mathematically, for small values of $\lvert \Delta \rvert$, the matrix $\mathbf{V} ( \theta^{\text{i}}, \theta^{\text{i}}+\Delta )$ becomes ill-conditioned, which in turn leads to ill-conditioning of the measurement matrix $\mathbf{H} ( \theta^{\text{i}}, \theta^{\text{i}}+\Delta )$. Consequently, even a small noise in the measurement procedure can result in much larger errors in the recovery. Additionally, the recovery performance worsens as the scanned range moves away from the normal direction of the RIS.
\end{rem}
 
\begin{rem}
  The third category that affects sensing performance is related to the number of measurements. The relationship between the bound and the number of measurements $T$ is described by the term $( 1 - \sqrt{ N / T } )^{-1}$. This indicates that the bound is inversely proportional to $1 - \sqrt{ N / T }$. When $T \gg N$, this term asymptotically approaches 1, its minimum value. As a rule of thumb, increasing the number of measurements well beyond the number of elements can potentially enhance the performance of backward sensing.
\end{rem}

\section{Numerical Validation} \label{S:Experiments}

Theorems~\ref{T:1} and \ref{T:RelativeError} provide theoretical insights into the key parameters that affect the performance of the proposed RIS-centric sensing approaches. In this section, numerical validations are presented. Specifically, our goal is to sense a region where multiple power sources of various shapes have been manually placed. We comprehensively analyze the impact of several factors, including the total number of samples, the total number of elements, the number of RISs and their positions. 

To evaluate the quality of backward sensing, we use two metrics: relative error and structural similarity index measure (SSIM). The SSIM is a perceptual metric to quantify the degradation of the recovered field $\hat {\mathbf{E}}$ compared to the ground truth $\mathbf{E}$, defined as
\[
  \text{SSIM}=\frac{ \left( 2 \mu_{ \mathbf{E} } \mu_{ \hat {\mathbf{E}} } + c_1 \right) \left( 2 \sigma_{ \mathbf{E} \hat {\mathbf{E}}} + c_2 \right) } { \left( \mu_{\mathbf{E}}^{2}+\mu_{ \hat {\mathbf{E}} }^{2}+c_1\right) \left( \sigma_{\mathbf{E}}^{2}+\sigma_{ \hat {\mathbf{E}} }^{2}+c_2 \right) } .
\]
Here $\mu_{ \hat {\mathbf{E}} }$ and $\mu_{ \mathbf{E} }$ are the means of $\hat {\mathbf{E}}$ and $\mathbf{E}$, $\sigma^2_{\hat {\mathbf{E}}}$ and $\sigma^2_{\mathbf{E}}$ are the variances, and $\sigma_{ \mathbf{E} \hat {\mathbf{E}}}$ is the covariance between $\mathbf{E}$ and $\hat {\mathbf{E}}$. The constants $c_1$ and $c_2$ are used to stabilize the division with weak denominator. The SSIM value lies within the range of $[0, 1]$, where a higher value indicates that the recovered field exhibits structures more similar to the ground truth.

Theorem~\ref{T:1} asserts that the rank of the sensing matrix is upper-bounded by the total number of measurements and the total number of elements. To assess the influence of these factors, we first examine the performance of backward sensing in the noiseless case.

\begin{figure}[!htbp] 
  \centerline{\includegraphics[width=0.9\columnwidth]{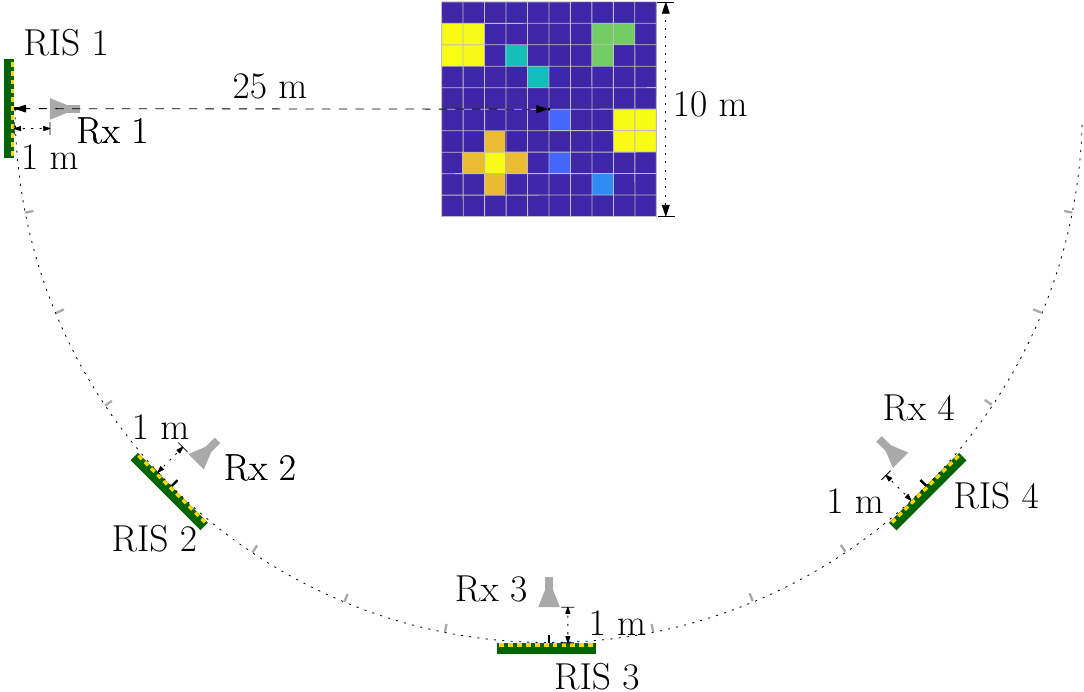}}
  \caption{Scenario configuration used in the simulation. A 10-meter by 10-meter square region is discretized into a 10 $\times$ 10 grid with 1-meter intervals. Four uniform linear RISs are strategically positioned 25 meters around the RoI, with each adjacent RIS separated by an angle of 45$^\circ$.}
  \label{T_scene} 
\end{figure}

\begin{table}[h!] 
  \centering
  \caption{Default setup parameters used in the simulation experiments.} 
  \label{tab:simulation_parameters} 
  \begin{tabular}{lll} 
  \toprule 
  \textbf{Parameter}                   & \textbf{Symbol}                                    &  \textbf{Value}     \\ \midrule 
  RoI                                  &   $x \times y $                                    &  10 m $\times$ 10 m   \\
  Spatial resolution (Pixel size)      &  $ \bigtriangleup  x \times \bigtriangleup y $ &   1 m $\times$ 1 m    \\
  Pixel dimension                      &  $ M_x \times M_y$                                 &   10 $\times$ 10    \\ 
  Number of RISs                        &  K                                                 &   4                 \\
  Operating frequency                  &  $f$                                               &   $30$ GHz          \\
  RIS element distance                 &  $d$                                               &   $0.01$ m          \\
  Distance from RoI to RIS             &  $L$                                               &   25 m               \\
  \bottomrule 
  \end{tabular} 
\end{table} 

Let's imagine a scenario where the task is to sense a square region measuring $10$ meters by $10$ meters. Four uniform linear RISs are strategically positioned around the RoI. Each RIS is equipped with a dedicated receiver. The topology and geometry are demonstrated in Fig.~\ref{T_scene}. The measurement vectors corresponding to each RIS are stacked into a column vector, as depicted in \eqref{E:Mode_1}. The system operates at a frequency of $f = 30$ GHz. The RoI is discretized into $10 \times 10$ pixels. The parameters are detailed in Table \ref{tab:simulation_parameters}. 

For a comprehensive illustration, we systematically record the recovered results as the number of measurements $T_k$ for each RIS varies at 10, 20, 50, and 100, with the number of elements fixed at $N_k = 50$ for each $k$. The total number of elements sums up to $N = \sum_{k=1}^4 N_k = 200$. The reconstructed scenes are illustrated in Fig.\ref{T_ls_matrix}. We plot the relative error and SSIM values in Fig.~\ref{SSIMofT}. Additionally, we increase the number of elements of each RIS from 20 to 50, while keeping the number of measurements fixed at $T_k = 50$ (yielding a total of $T = \sum_{k=1}^4 T_k = 200$ measurements). The results are presented in Figs.~\ref{N_ls_matrix} and \ref{SSIMofN}.

\begin{figure}[!htbp] 
  \centerline{\includegraphics[width=1\columnwidth]{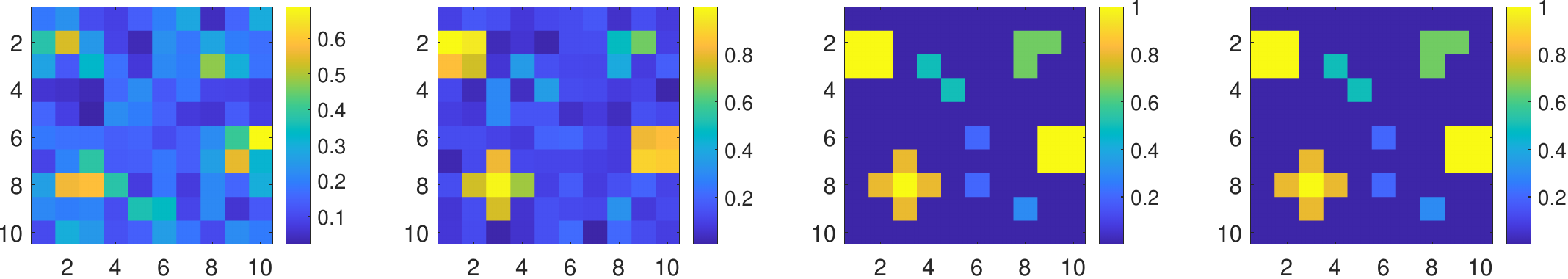}}
  \caption{Comparison of the recovered scenes as the number of measurements $T_k$ for each RIS varied at 10, 20, 50, and 100, with the number of elements fixed at $N_k = 50$ for each $k$.}
  \label{T_ls_matrix}
\end{figure}

\begin{figure}[!htbp] 
  \centerline{\includegraphics[width=0.6\columnwidth]{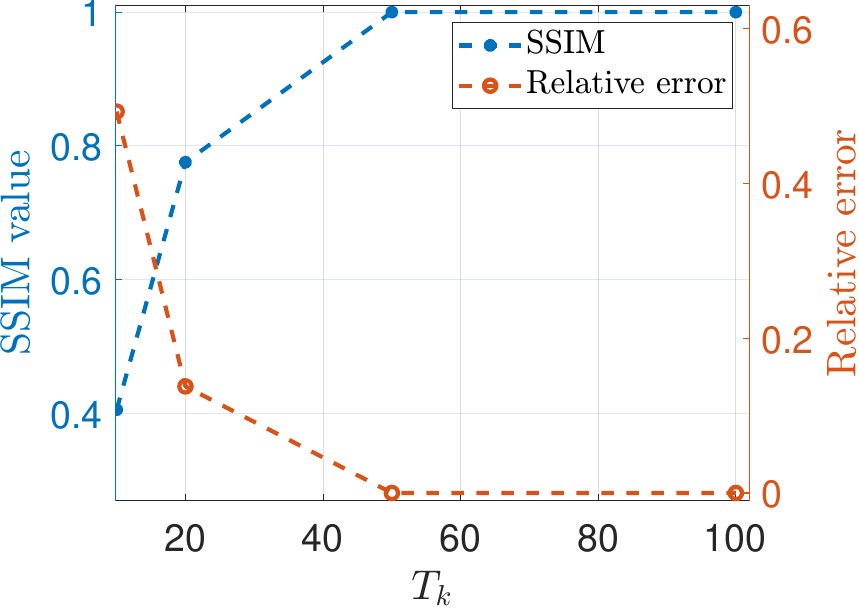}}
  \caption{Relative error and SSIM as the number of measurements $T_k$ for each RIS varied at 10, 20, 50, and 100, with the number of elements fixed at $N_k = 50$ for each $k$.}
  \label{SSIMofT}
\end{figure}

\begin{figure}[!htbp] 
  \centerline{\includegraphics[width=1\columnwidth]{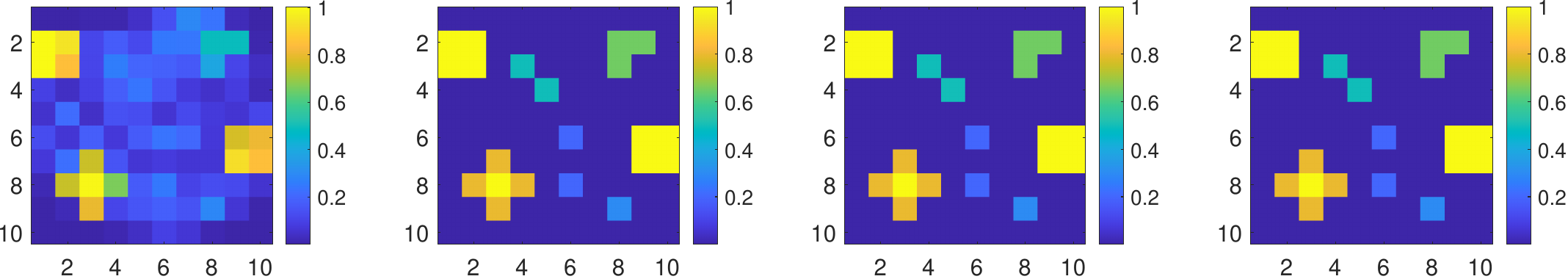}}
  \caption{Comparison of the recovered scenes as the number of elements $N_k$ for each RIS varied at 20, 30, 40, and 50, with the number of measurements fixed at $T_k = 50$ for each $k$.}
  \label{N_ls_matrix}
\end{figure}

\begin{figure}[!htbp] 
  \centerline{\includegraphics[width=0.6\columnwidth]{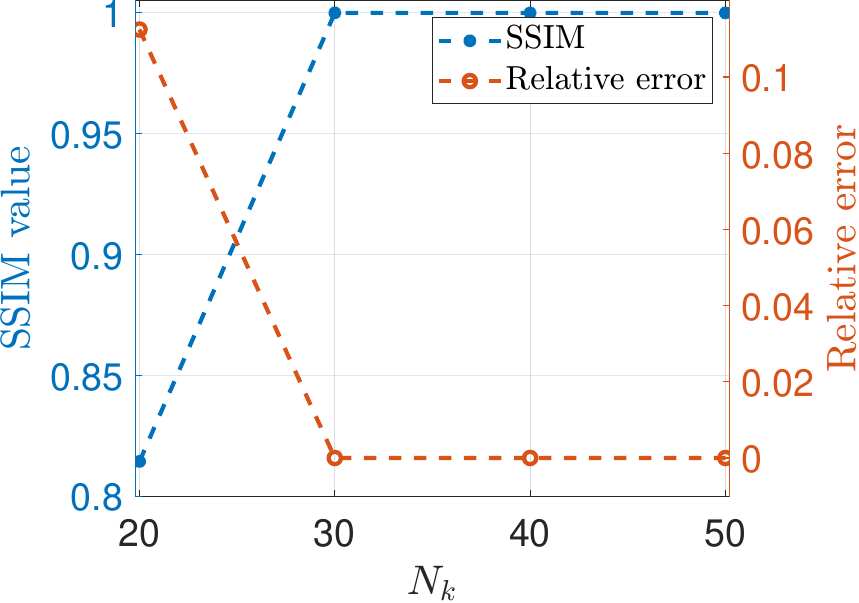}}
  \caption{Relative error and SSIM as the number of elements $N_k$ for each RIS varied at 20, 30, 40, and 50, with the number of measurements fixed at $T_k = 50$ for each $k$.}
  \label{SSIMofN}
\end{figure}

These plots demonstrate that increasing both the number of measurements and the number of elements enhances the performance of backward sensing. Thess results are consistent with intuitive expectations. When $N$ or $T$ is not sufficiently large, a fundamental limitation arises: the sensing matrix encodes the high-dimensional scene into a low-dimensional observation. Clearly, this dimensionality reduction makes recovering the original scene using linear operators impractical. For instance, with $T_k = 10$, accurately reconstructing an original vector in $\mathbb{C}^{100}$ from an observation lying in $\mathbb{C}^{40}$ becomes challenging.

It's worth emphasizing that even if the sensing matrix encodes the scene vector into a high-dimensional observation, achieving high-precision recovery might still be challenging. The presence of noise, particularly in cases of ill-conditioned sensing matrices, makes high-fidelity reconstruction from observations unfeasible. For better understanding, we illustrate the singular values corresponding to all the aforementioned trials in Fig.~\ref{SVD_Illustration}.

\begin{figure}[!htbp]
  \centering
  \subfigure[]{
  \label{SVD_1}
  \includegraphics[width=.48\columnwidth]{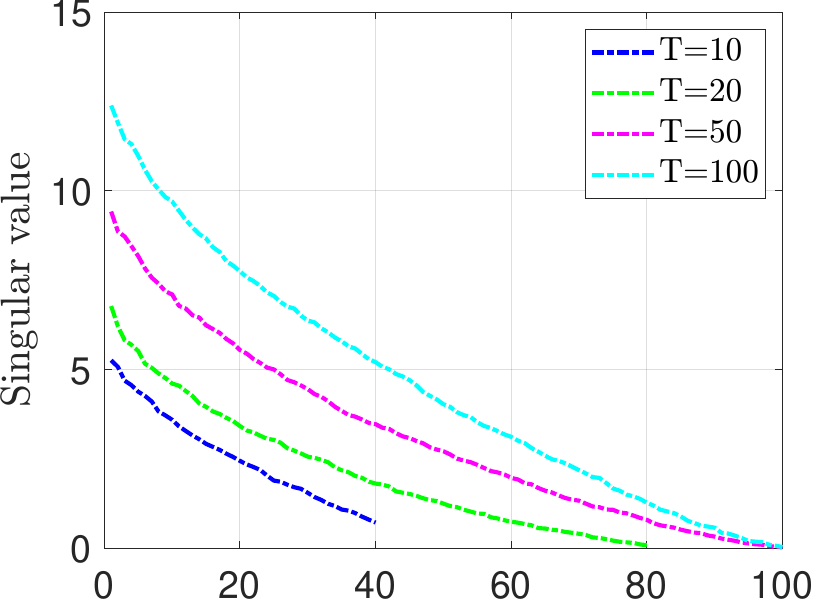}}
  \subfigure[]{
  \label{SVD_2}
  \includegraphics[width=.48\columnwidth]{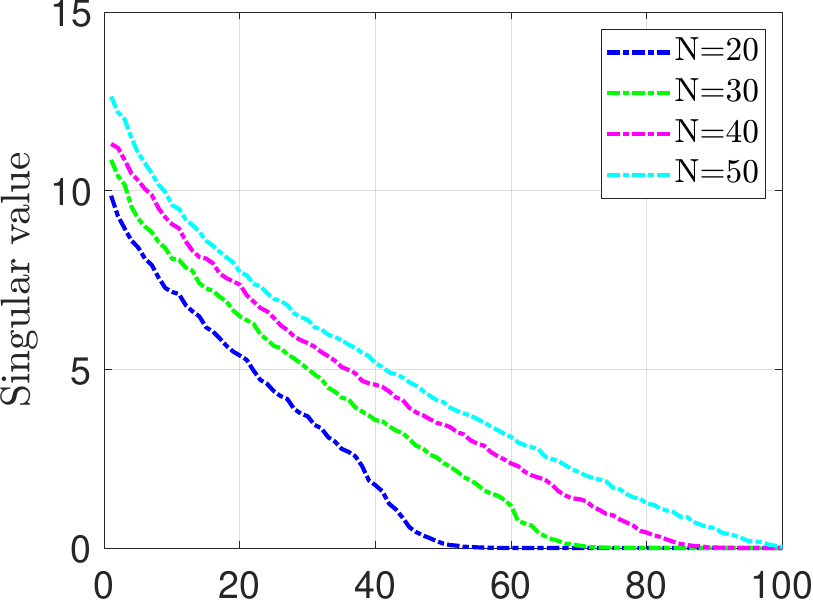}}
  \caption{Sorted singular values corresponding to all trials. (a) The number of measurements $T_k$ is varied at 10, 20, 50, and 100, with the number of elements fixed at $N_k = 50$. (b) $K$ The number of elements $N_k$ is varied at 20, 30, 40, and 50, with the number of measurements fixed at $T_k = 50$.}
  \label{SVD_Illustration} 
\end{figure}

In the context of backward sensing using multiple RISs, the conditioning of the sensing matrix depends not only on the number of elements and measurements but also on factors such as the system's geometry and topology. Optimizing the topology of RISs has the potential to enhance the conditioning of the sensing matrix, thereby improving the robustness of scene reconstruction.

To demonstrate this improvement, we conduct a series of experiments by strategically selecting eight landmark positions around the RoI, each of which serves as a potential location for an RIS. These RISs are positioned tangentially to a circle encompassing the RoI. The location of these landmarks are shown in Fig~\ref{PK_scene}. The specific positions are detailed in Table~\ref{tab:KP_parameters}. To ensure a fair comparison, we maintain the total number of elements at $N = \sum_{k} N_k = 200$ and the total number of measurements at $T = \sum_{k} T_k = 200$. To assess the stability of the sensing system, additive Gaussian noise is introduced in these experiments, with the SNR set to 30 dB.

\begin{table}
  \centering
  \caption{Deployment strategies for combinations of different candidate positions.} 
  \label{tab:KP_parameters} 
  \begin{tabular}{lll} 
  \toprule 
  \textbf{Category}                        & \textbf{Configuration}              &  \textbf{Description}                                            \\ \midrule 
  Strategy \uppercase\expandafter{\romannumeral1}:     &  \{A\}                           &  \makecell[l]{-Use only a single landmark position }                 \\
  Strategy \uppercase\expandafter{\romannumeral2}:    &  \{A, E\}                        &  \makecell[l]{-Use two landmark positions perpendi-\\cular to the RoI}  \\
  Strategy \uppercase\expandafter{\romannumeral3}:  &  \{A, C, E, G\}                  &  \makecell[l]{-Use four landmark positions with each \\ adjacent interval of them at an angle of \\45$^\circ$} \\ 
  Strategy \uppercase\expandafter{\romannumeral4}:  & \makecell[l]{\{A, B, C, D,\\ E, F, G, H\} } &  \makecell[l]{-Use all landmark positions with high \\landmark density}      \\ 
  \bottomrule 
  \end{tabular} 
\end{table} 

The results of the backward sensing experiments are presented in Fig.\ref{KP_ls_matrix}. In addition to the reconstructed scenes, we also plot the sorted singular values and the condition numbers of the sensing matrices in Fig.\ref{KP_cond_number_SVD}. For better illustration, the y-axis uses a logarithmic scale. It is evident that recovering the scene with acceptable fidelity is impossible when using only 1 or 2 RISs in this specific scenario. Notably, when using 2 RISs, the condition number of the sensing matrix is greater than $3.1 \times 10^5$. The backward sensing procedure is highly sensitive to perturbations in the observations.

\begin{figure}[!htbp] 
  \centerline{\includegraphics[width=1\columnwidth]{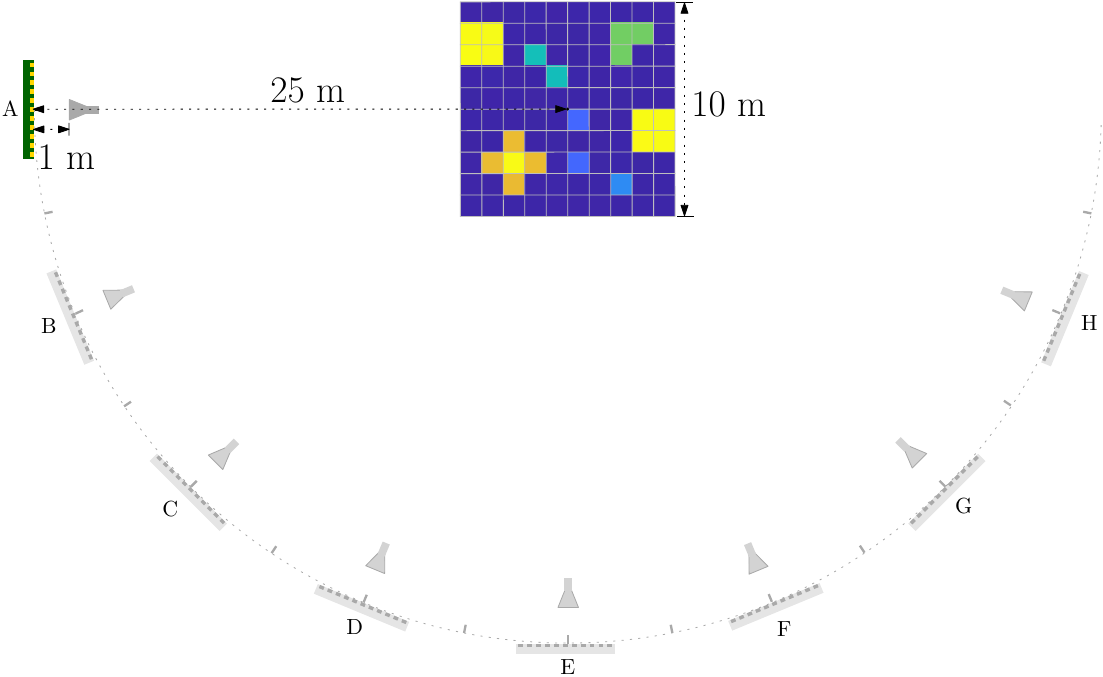}}
  \caption{Scenario configuration and landmark positions. Eight landmark positions are set 25 meters around the RoI, labeled from A to H. These landmarks are spaced at 22.5$^\circ$ intervals relative to the RoI, serving as candidate locations for RIS deployment.}
  \label{PK_scene}
\end{figure}

\begin{figure}[!htbp] 
  \centerline{\includegraphics[width=0.9\columnwidth]{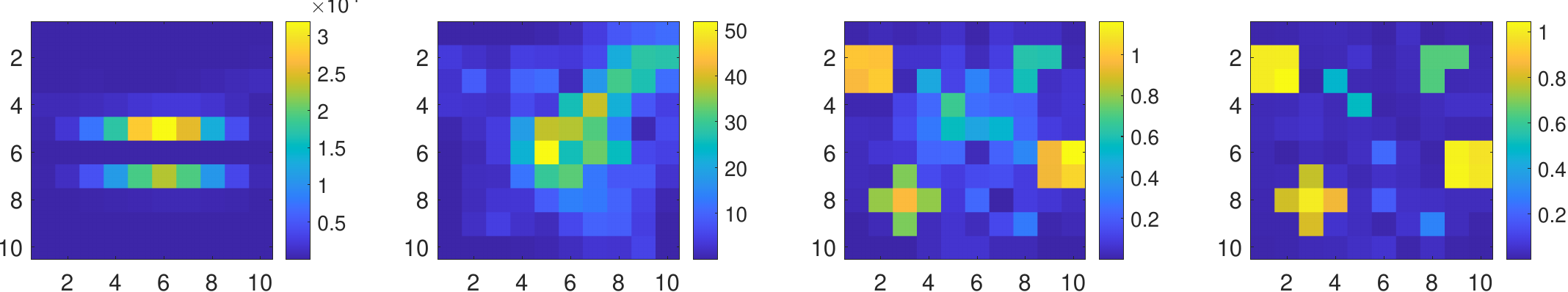}} 
  \caption{Comparison of the recovered scenes using various geometries and topologies of RISs. Achieving acceptable fidelity in scene recovery is impossible when employing only 1 or 2 RISs in this specific scenario.}
  \label{KP_ls_matrix} 
\end{figure}

\begin{figure}[!htbp] 
  \centering
  \subfigure[]{
  \label{SVSofsensingmatrixKP}
  \includegraphics[width=.47\columnwidth]{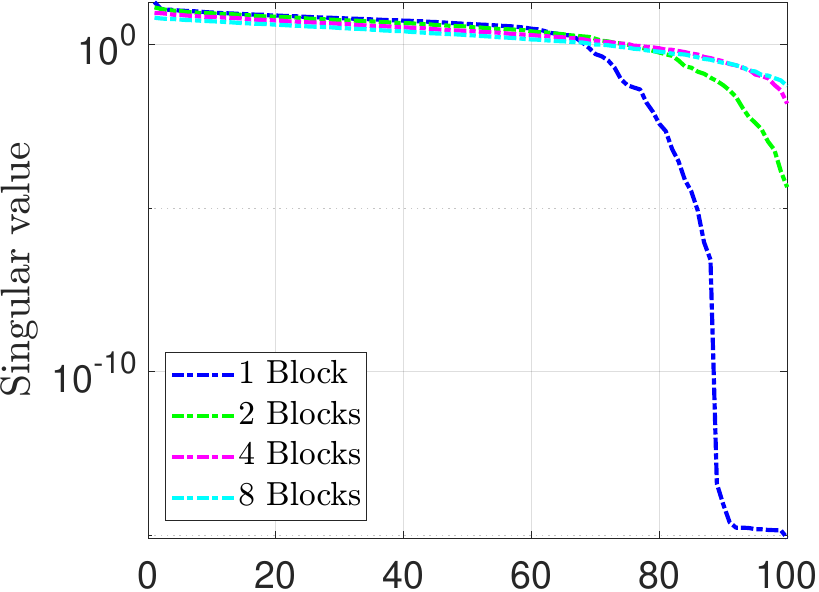}}
  \subfigure[]{
  \label{KP_cond_number}
  \includegraphics[width=.465\columnwidth]{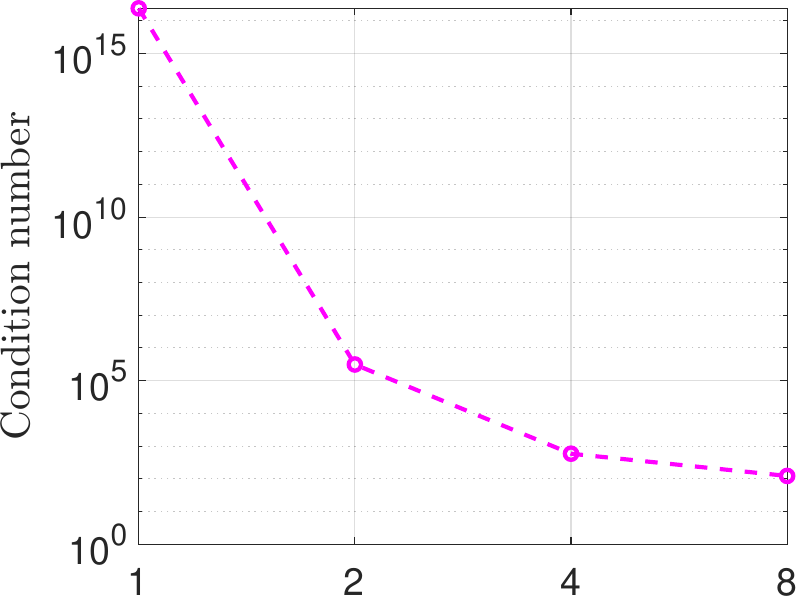}}
  \caption{(a) Sorted singular values for each deployment strategy, with the y-axis presented on a logarithmic scale for clarity. (b) Condition numbers of the measurement matrix for each deployment strategy.}
  \label{KP_cond_number_SVD}
\end{figure}

\begin{figure}
  \centering
  \subfigure[]{
  \label{fig:single board1}
  \includegraphics[width=.465\columnwidth]{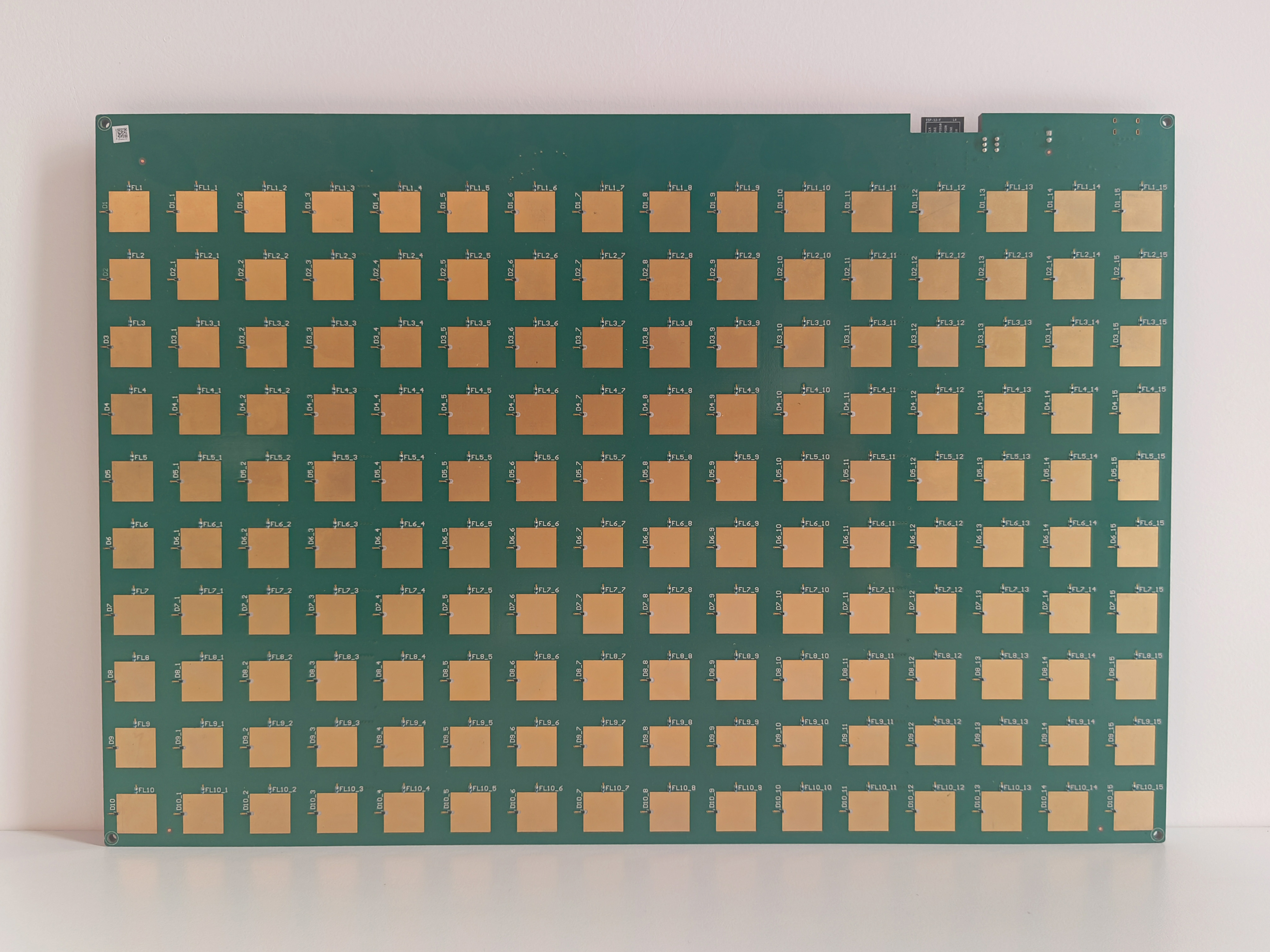}}
  \subfigure[]{
  \label{fig:single board2}
  \includegraphics[width=.465\columnwidth]{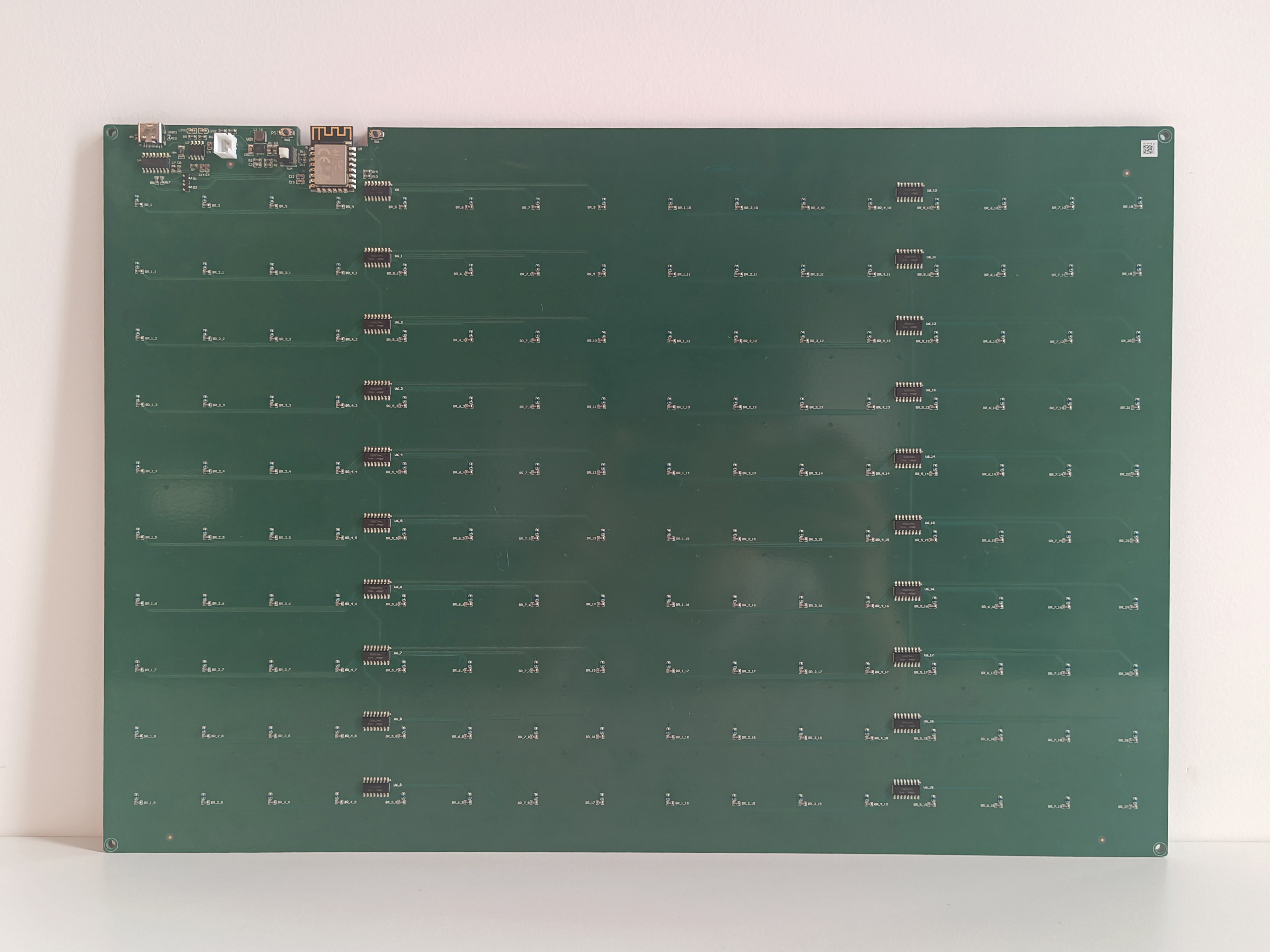}}
  \caption{Photographs of the front (a) and back (b) of the fabricated RIS prototype with 10$\times$ 16 elements. The RIS, operating at a frequency of 5.8 GHz, consist of 10 $\times$ 16 1-bit elements with an element spacing of 0.025 m.}
  \label{fig:single board}
\end{figure}

\begin{figure}
  \centerline{\includegraphics[width=0.95\columnwidth]{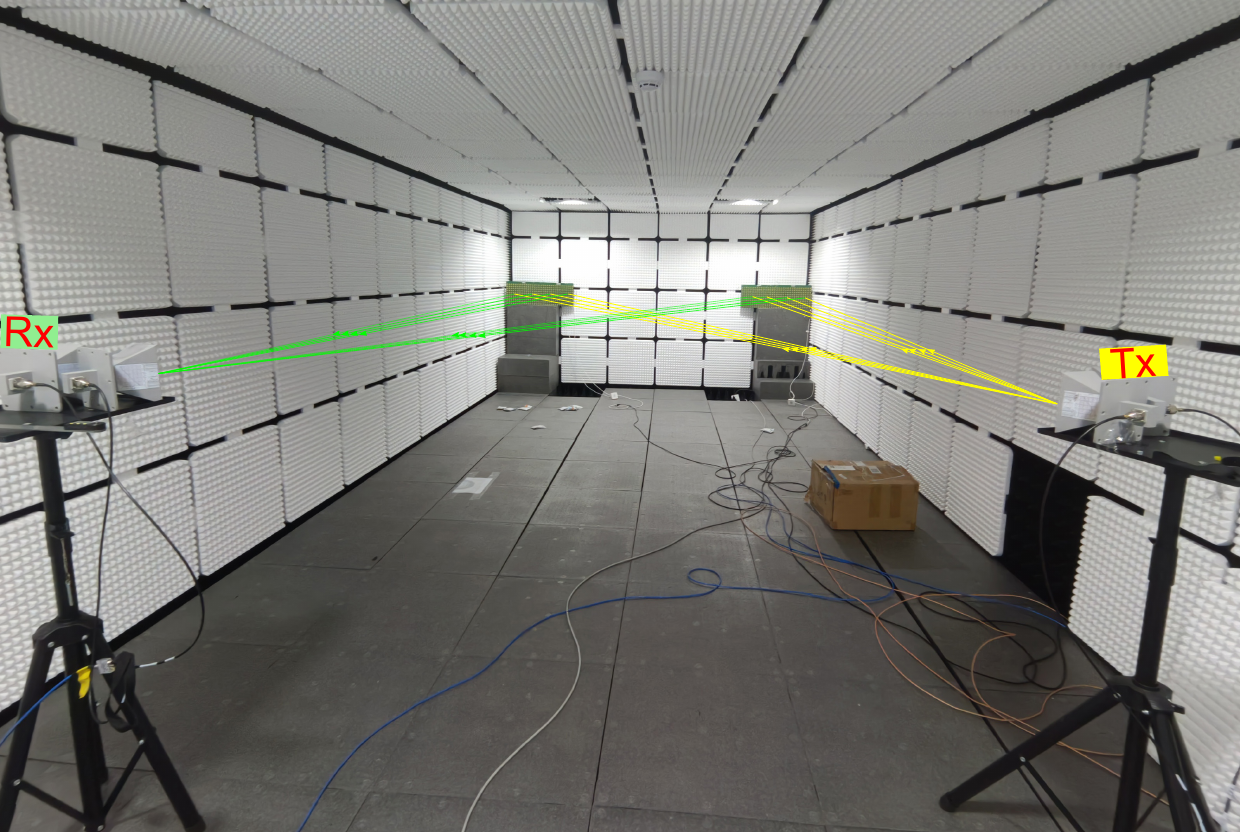}}
  \caption{Experimental setup of two RISs sensing system in a microwave anechoic chamber.}
  \label{fig:experiment_scenario}
\end{figure}

\section{Proof-of-concept Prototype} \label{S:POC}
To validate the proposed sensing scheme, we implement a proof-of-concept prototype using a universal software radio peripheral (USRP) in a microwave anechoic chamber. Two RISs are strategically positioned at one end of the chamber to detect a power source located at the other end, as illustrated in Fig.~\ref{fig:experiment_scenario}. The schematic is depicted as Fig.~\ref{fig:two_RISs Schematic}. The RISs, operating at a frequency of 5.8 GHz, consist of 10 $\times$ 16 1-bit elements with an element spacing of 0.025 m, as shown in Fig.~\ref{fig:single board}. To evaluate the performance of uniform linear RISs within a planar RoI, each column is configured identically. The two RISs, separated by 2.81 m, are employed to detect the source positioned 6 m away. For each RIS, $T=500$ measurements are acquired to construct the sensing matrix $\mathbf{H}$.

\begin{figure}
  \centerline{\includegraphics[width=0.92\columnwidth]{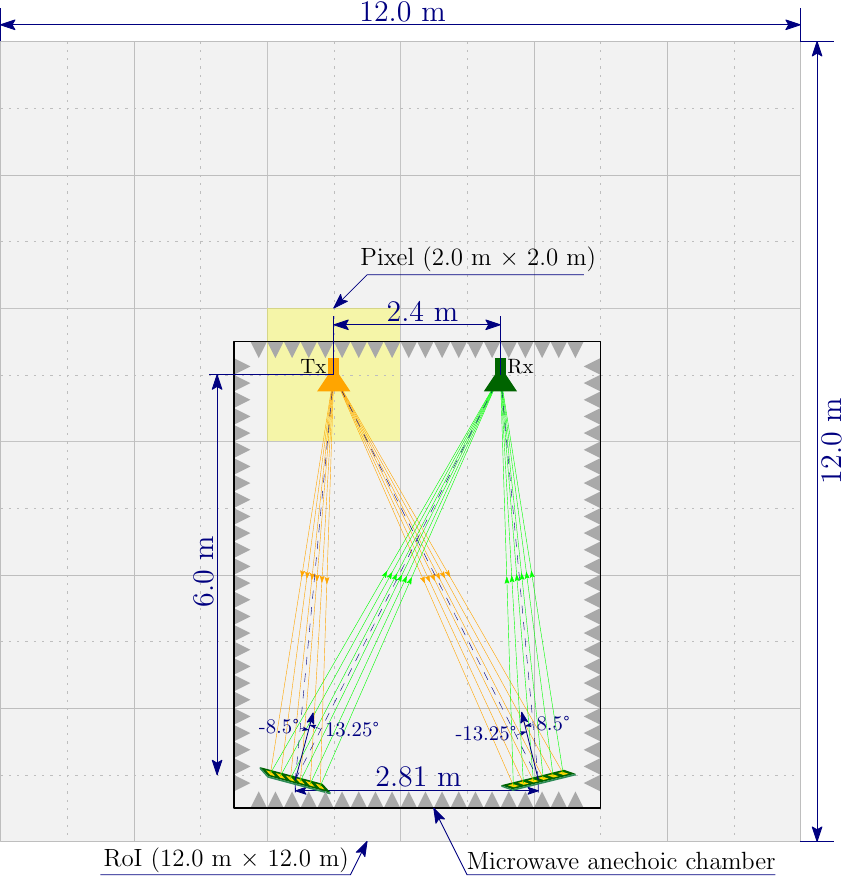}}
  \caption{The geometry and topology of the experiment conducted inside the anechoic chamber. The RoI is a square with dimensions of 12 m $\times$ 12 m, discretized into a grid of 6 $\times$ 6, resulting in each pixel being 2 m $\times$ 2 m.}
  \label{fig:two_RISs Schematic}
\end{figure}

\begin{figure}[!htbp]
  \centerline{\includegraphics[width=0.9\columnwidth]{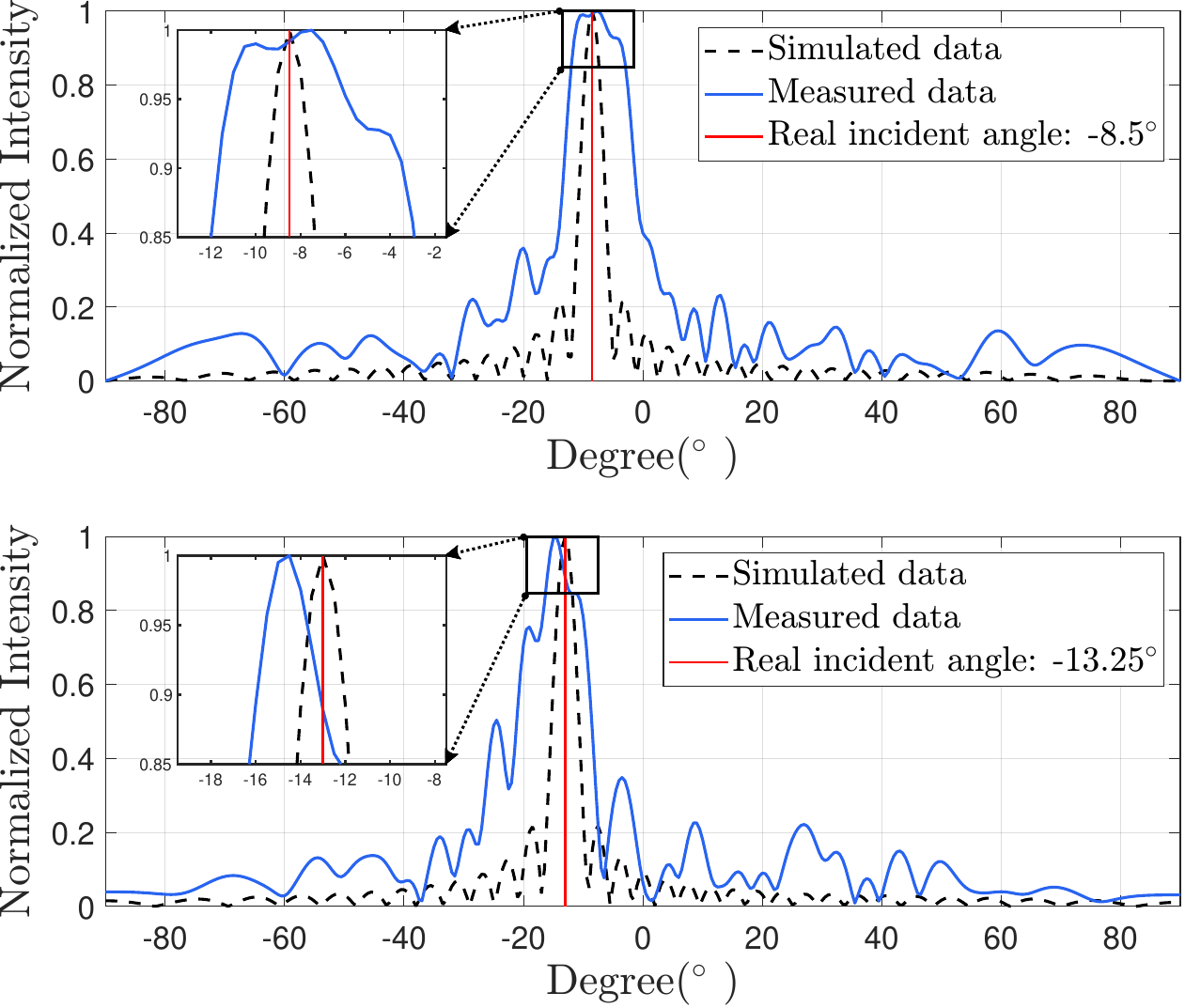}}
  \caption{Results of DoA derived from both simulated and measured data by the magnitude-only reconstruction algorithm. The first row displays the results of DoA from the left RIS. The second row displays the results of DoA from the right RIS. Results demonstrate high accuracy with errors within $2^\circ$.}
  \label{1simulation}
\end{figure}

\begin{figure}[htbp] 
  \centerline{\includegraphics[width=0.85\columnwidth]{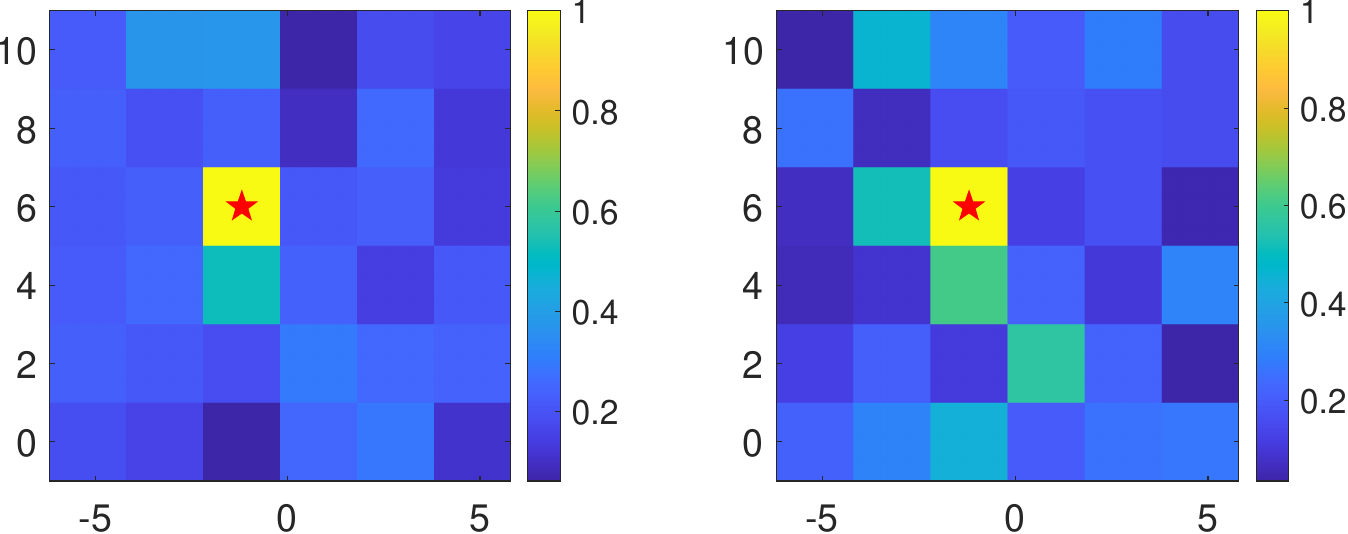}}
  \caption{Recovered results using the magnitude-only reconstruction algorithm in the chamber. The left figure displays the reconstructed scene from simulated data. The right figure shows the reconstructed scene from measured data.} 
  \label{experiment_Reconstructed_real1}
\end{figure}

It's important to note that the validation using USRP has a limitation: while the strength of the carrier wave is easy to measure, accurately capturing the phase is more challenging. Therefore, the backward sensing process must rely on magnitude-only reconstruction algorithms. 
Mathematically, the magnitude-only backward sensing problem is formulated as follows: given $\mathbf{H}$ and $\lvert \mathbf{S} \rvert$, find $\mathbf{E}$ such that
\[
  \lvert \mathbf{S} \rvert \approx \lvert \mathbf{H} \mathbf{E} \rvert .
\]
In contrast to phased measurements described by \eqref{E:LinearMeasurement}, the missing phase makes the problem more complex. The phaseless reconstruction problem is highly nonlinear and significantly more ill-posed than the phased reconstruction problem~\cite{ammari2016phased}. Various approaches have been developed to address the phaseless reconstruction problem, including gradient descent methods, iterative projection methods, and convex optimization-based methods~\cite{fannjiang2020numerics}. We focuse on validating our backward sensing approach rather than delving into the specifics of these algorithms. For our purposes, we employ the reweighted Wirtinger flow (RWF) algorithm as the solver~\cite{yuan2017phase}.

We initiate the validation of the RWF algorithm by estimating the DoA using data independently obtained from each of the two RISs. Fig.~\ref{1simulation} shows the DoA estimations from both RISs. Sepcifically, a single source is localized at $-8.5^\circ$ for the left RIS and $-13.25^\circ$ for the right RIS. In numerical experiments, a distinct peak is observed with DoA estimations of $-8.5^\circ$ and $-13^\circ$ for the left RIS and the right RIS, respectively. In the results estimated using measured data, the DoA estimations are $-7.5^\circ$ for the left RIS and $-14.5^\circ$ for the right RIS, with errors within two degrees. These figures illustrate the remarkably accurate estimation of incident angles, thereby demonstrating the accuracy of our modeling and the effectiveness of the RWF algorithm within this experimental setup.

To further validate the collaborative sensing capabilities of the two RISs, we define the RoI as a square with dimensions of 12 m $\times$ 12 m, discretized into a grid of 6 $\times$ 6, resulting in each pixel being 2 m $\times$ 2 m. This choice of resolution takes into account the limited separation between the two RISs and the relatively small size of each RIS, which comprises only 16 elements. The final results, illustrated in Fig.~\ref{experiment_Reconstructed_real1}, provide a comprehensive overview. The true location of the source is indicated by red pentagrams. The plots clearly demonstrate that the power source can be precisely located using the two RISs. These findings validate the feasibility of the proposed backward sensing approach with multiple RISs.


\begin{figure*}
  \begin{equation} \label{covariance} 
    \begin{aligned}
      {\mathbf{V}}^* \mathbf{V} 
      &= \begin{bmatrix}
        1 & e^{ -j 2 \pi d \sin \theta^{\text{i}} / \lambda }  & \cdots & e^{ -j 2 \pi (N-1) d \sin \theta^{\text{i}} / \lambda }\\
        1 & e^{ -j 2 \pi d \sin ( \theta^{\text{i}}+\Delta ) / \lambda }  & \cdots & e^{ -j 2 \pi (N-1) d \sin ( \theta^{\text{i}}+\Delta ) / \lambda } \\
      \end{bmatrix} 
      \begin{bmatrix}
          1      & 1 \\
          e^{ j 2 \pi d \sin \theta^{\text{i}} / \lambda }    & e^{ j 2 \pi d \sin ( \theta^{\text{i}}+\Delta ) / \lambda }\\
          \vdots & \vdots\\
          e^{ j 2 \pi (N-1) d \sin \theta^{\text{i}}/ \lambda }  & e^{ j 2 \pi (N-1) d \sin ( \theta^{\text{i}}+\Delta ) / \lambda }
        \end{bmatrix}\\
    &= \begin{bmatrix}
      N       & 1  + \cdots + e^{ j 2 \pi (N-1) d  (\sin ( \theta^{\text{i}}+\Delta ) -\sin \theta^{\text{i}})/ \lambda }\\
      1  + \cdots + e^{ j 2 \pi (N-1) d (\sin \theta^{\text{i}}- \sin ( \theta^{\text{i}}+\Delta )) / \lambda } & N
    \end{bmatrix} .
    \end{aligned}
  \end{equation}
  
  \begin{equation} \label{lambda1}
    \begin{aligned}
      \lambda_1,\lambda_2  
      & = N \pm \sqrt{(1 + \dots + e^{ j 2 \pi (N-1) d (\sin \theta^{\text{i}}- \sin ( \theta^{\text{i}}+\Delta )) / \lambda } ) (1+\dots+e^{ j 2 \pi (N-1) d  (\sin ( \theta^{\text{i}}+\Delta ) -\sin \theta^{\text{i}})/ \lambda })}\\
      & = N \pm \left|\frac{1-e^{j 2 \pi N d (\sin \theta^{\text{i}}- \sin ( \theta^{\text{i}}+\Delta ))\lambda}}{1-e^{j 2 \pi d (\sin \theta^{\text{i}}- \sin ( \theta^{\text{i}}+\Delta )) / \lambda}} \right|, \\
      & = N \pm \left| \frac{ \sin ( \pi N d (\sin \theta^{\text{i}} - \sin ( \theta^{\text{i}}+\Delta ) ) / \lambda ) }{ \sin ( \pi d ( \sin \theta^{\text{i}} - \sin ( \theta^{\text{i}}+\Delta ) ) / \lambda ) } \right| .\\
    \end{aligned}
  \end{equation}
  \medskip
  \hrule
\end{figure*}

\section{Conclusion} \label{S:Conclusion}

This paper investigates the capability and practicality of the proposed RIS-centric sensing approach. We present physically accurate and mathematically concise models to characterize the reflection properties of RISs. Using linearized approximations inherent in the far-field region, practical models are derived to describe the forward signal aggregation process via RISs. We demonstrate that DoA estimation of incident waves can be achieved with a single RIS. Additionally, we show that the spatial diversity provided by deploying multiple RISs allows for accurate localization of multiple power sources. A theoretical framework for determining key performance indicators is established through condition number analysis of the sensing operators. Numerical experiments validate our findings. To showcase the practicality of our proposed RIS-centric sensing approach, we develop a proof-of-concept using USRP. To the best of our knowledge, no previous prototypes have been proposed with such configurations.

\appendices

\section{Proof of Lemma~\ref{L:4}}\label{S:Appendix_A3}
\begin{proof} 
For the matrix $\mathbf{V}$ defined by \eqref{E:VandermodeMatrix}, we first calculate ${\mathbf{V}}^* \mathbf{V}$, as shown in 
\eqref{covariance}. It is straightforward to calculate the eigenvalues of ${\mathbf{V}}^* \mathbf{V}$, which are given in \eqref{lambda1}. Finally, we establish the singular values of $\mathbf{V}$. 
\end{proof}

\section{Proof of Lemma~\ref{L:5}}\label{S:Appendix_A4}
\begin{proof} 
The main idea of the proof is to find the series expansion
\begin{equation} \label{PL5:1}
  \frac{\sin N x}{\sin x} = \sum_{k=0}^{\infty} a_{k} x^{2 k} .
\end{equation}
It is easy to verify that
\begin{equation} \label{D:4}
  \begin{aligned}
    \frac{\sin N x}{\sin x}
    &= \frac{e^{i N x}-e^{-i N x}}{e^{i x}-e^{-i x}}
    =\sum_{l=0}^{N-1} e^{i(N-1-2 l) x}\\
    &=\left\{\begin{array}{ll}
      1+2 \sum_{l=1}^{\frac{N - 1}{2}} \cos 2 l x  & N \text { odd}, \\
      2 \sum_{l=1}^{\frac{N}{2}} \cos (2 l -1) x  & N \text { even}.
    \end{array} \right.
  \end{aligned}
\end{equation}
Using the Taylor expansion $\cos x = \sum_{k=0}^{\infty} \frac{(-1)^k}{(2k)!} x^{2 k}$, we find 
\begin{equation}\label{E:TaylorCos1}
  \cos 2 l x = \sum_{k=0}^{\infty} \frac{(-1)^k}{(2k)!} (2 l)^k x^{2 k} ,
\end{equation}
and
\begin{equation}\label{E:TaylorCos2}
  \cos (2 l -1) x = \sum_{k=0}^{\infty} \frac{(-1)^k}{(2k)!} (2 l - 1)^k x^{2 k} .
\end{equation}
Substituting \eqref{E:TaylorCos1} and \eqref{E:TaylorCos2} into \eqref{D:4} yields
\begin{equation}
    a_{k} = \frac{(-1)^{k}}{(2 k) !} \sum_{l=0}^{N-1}(N-1-2 l)^{2 k} .
\end{equation}
It follows that $a_0 = N$ and $a_1 = -\frac{N (N^2-1)}{6}$. Finally, for sufficiently small $x$, we have 
\begin{equation}
  \frac{ \sin (N x) }{\sin x} \approx N - \frac{ N ( N^2-1 ) x^2 }{ 6 }. 
\end{equation}
\end{proof} 

\bibliographystyle{IEEEtran}
\bibliography{Reference}

\end{document}